\documentclass[3p,11pt]{elsarticle}

\newcommand{\red}[1]{{#1}\xspace}
\newcommand{\markRed}{}

\usepackage{latexsym}
\usepackage{xspace}
\usepackage{amssymb}
\usepackage{stmaryrd}
\usepackage{hyperref}
\usepackage{url}
\usepackage{graphicx}
\input{xy}
\xyoption{all}
\usepackage{calc}
\usepackage{tikz}
\usepackage{flushend}
\usepackage{color}

\newtheorem{theorem}{Theorem}[section]
\newtheorem{lemma}[theorem]{Lemma}
\newtheorem{proposition}[theorem]{Proposition}
\newtheorem{corollary}[theorem]{Corollary}

\newtheorem{Definition}[theorem]{Definition}
\newtheorem{Example}[theorem]{Example}
\newtheorem{Remark}[theorem]{Remark}

\newenvironment{definition}{\begin{Definition}\begin{em}}{\end{em}\end{Definition}}
\newenvironment{example}{\begin{Example}\begin{em}}{\end{em}\end{Example}}
\newenvironment{remark}{\begin{Remark}\begin{em}}{\end{em}\end{Remark}}
\newproof{proof}{Proof}

\def\eqref#1{(\ref{#1})}

\def\tuple#1{\langle#1\rangle}

\newcommand{\E}{\exists}

\newcommand{\mL}{\mathcal{L}}

\newcommand{\mA}{\mathcal{A}}
\newcommand{\mAp}{{\mathcal{A}'\!}}

\newcommand{\mZ}{\mathcal{Z}}

\newcommand{\mM}{\mathcal{M}}
\newcommand{\mMp}{{\mathcal{M}'\!}}

\newcommand{\NN}{\mathbb{N}}

\newcommand{\myend}{\mbox{}\hfill{\scriptsize$\blacksquare$}}
\newcommand{\Myend}{\mbox{}\hfill{_\blacksquare}}

\newcommand{\comment}[1]{}

\newcommand{\fand}{\varotimes}

\newcommand{\fto}{\Rightarrow}
\newcommand{\fequiv}{\Leftrightarrow}

\newcommand{\FLTS}{FLTS\xspace}
\newcommand{\FLTSs}{FLTSs\xspace}

\newcommand{\SP}{\Sigma_P}
\newcommand{\SA}{\Sigma_A}

\newcommand{\fPDLwP}{$\mathit{fPDL}^{-\Phi}$\xspace}
\newcommand{\fPDL}{$\mathit{fPDL}$\xspace}
\newcommand{\fKz}{$\mathit{fK}^0$\xspace}

\journal{arXiv}

\begin{document}
\sloppy
	
\begin{frontmatter}
		
\title{Logical Characterizations of Fuzzy Bisimulations in Fuzzy Modal Logics over Residuated Lattices}

\author{Linh Anh Nguyen}
\ead{nguyen@mimuw.edu.pl}

\address{
	Institute of Informatics, University of Warsaw, 
	Banacha 2, 02-097 Warsaw, Poland 
}

\address{
	Faculty of Information Technology, Nguyen Tat Thanh University, Ho Chi Minh City, Vietnam
}

\begin{abstract}
There are two kinds of bisimulation, namely {\em crisp} and {\em fuzzy}, between fuzzy structures such as fuzzy automata, fuzzy labeled transition systems, fuzzy Kripke models and fuzzy interpretations in description logics. Fuzzy bisimulations between fuzzy automata over a complete residuated lattice have been introduced by {\'C}iri{\'c} {\em et al}.~in 2012. Logical characterizations of fuzzy bisimulations between fuzzy Kripke models (respectively, fuzzy interpretations in description logics) over the residuated lattice $[0,1]$ with the G\"odel t-norm have been provided by Fan in~2015 (respectively, Nguyen {\em et al}.~in 2020). There was the lack of logical characterizations of fuzzy bisimulations between fuzzy graph-based structures over a general residuated lattice, as well as over the residuated lattice $[0,1]$ with the {\L}ukasiewicz or product t-norm. In this article, we provide and prove logical characterizations of fuzzy bisimulations in fuzzy modal logics over residuated lattices. The considered logics are the fuzzy propositional dynamic logic and its fragments. Our logical characterizations concern invariance of formulas under fuzzy bisimulations and the Hennessy-Milner property of fuzzy bisimulations. They can be reformulated for other fuzzy structures such as fuzzy labeled transition systems and fuzzy interpretations in description logics. 
\end{abstract}


\end{frontmatter}


\section{Introduction}
\label{section:intro}

Bisimulation is a useful notion for characterizing equivalence of states in transition systems~\cite{Park81,HennessyM85}. It has been extensively studied in modal logic for characterizing logical indiscernibility of states and separating the expressive power of modal logics (see, e.g., \cite{vBenthem84,BRV2001,GorankoOtto06}). Bisimulation can be used for minimizing structures. It can also be exploited for studying indiscernibility of individuals and concept learning in description logics~\cite{LbRoughification}. 

To deal with vagueness, fuzzy structures are used instead of crisp ones. There are two kinds of bisimulation, namely {\em crisp} and {\em fuzzy}, between fuzzy structures such as fuzzy automata, fuzzy transition systems, fuzzy Kripke models and fuzzy interpretations in description logics. 
Researchers have studied crisp bisimulations for fuzzy transition systems~\cite{CaoCK11,CaoSWC13,DBLP:journals/fss/WuD16,DBLP:journals/ijar/WuCHC18,DBLP:journals/fss/WuCBD18}, weighted automata~\cite{DamljanovicCI14}, Heyting-valued modal logics~\cite{EleftheriouKN12}, G\"odel modal logics~\cite{Fan15} and fuzzy description logics~\cite{FSS2020}. 
They have also studied fuzzy bisimulations for fuzzy automata~\cite{CiricIDB12,CiricIJD12}, weighted/fuzzy social networks~\cite{ai/FanL14,IgnjatovicCS15}, G\"odel modal logics~\cite{Fan15} and fuzzy description logics~\cite{FSS2020,minimization-by-fBS,TFS2020}. 

This article is devoted to \red{study} logical characterizations of fuzzy bisimulations. 
In~\cite{Fan15} Fan introduced fuzzy bisimulations between fuzzy Kripke models over the lattice $[0,1]$ using the G\"odel semantics (i.e., the G\"odel t-norm and its residuum). She provided logical characterizations of such bisimulations in the basic fuzzy monomodal logic and its extension with converse. The results concern invariance of modal formulas under fuzzy bisimulations and the Hennessy-Milner property of fuzzy bisimulations. 
In \cite{ai/FanL14} Fan and Liau studied fuzzy bisimulations under the name ``regular equivalence relations'' for weighted social networks. They provided logical characterizations for such bisimulations under the G\"odel semantics, including invariance results and the Hennessy-Milner property. 
In~\cite{FSS2020} Nguyen {\em et al}.\ defined and studied fuzzy bisimulations for a large class of fuzzy description logics under the G{\"o}del semantics. The work~\cite{FSS2020} contains results on invariance of concepts under such bisimulations and the Hennessy-Milner property of such bisimulations. 
In~\cite{EleftheriouKN12} Eleftheriou {\em et al}.\ studied bisimulations for Heyting-valued modal logics. Bisimulations defined in~\cite{EleftheriouKN12} are crisp and cut-based (i.e., using fuzzy values as thresholds). As discussed by Fan~\cite{Fan15}, such bisimulations give another representation of fuzzy bisimulations. The work~\cite{EleftheriouKN12} contains results on logical characterizations of the studied bisimulations, including the Hennessy-Milner property.

Note that the results on fuzzy bisimulations of all the works~\cite{ai/FanL14,Fan15,FSS2020} are formulated and proved only for fuzzy structures over the lattice $[0,1]$ using the G\"odel semantics. The results of~\cite{EleftheriouKN12} concern only modal logics over Heyting algebras of truth values. Such algebras are residuated lattices that use $\fand = \land$, and therefore, are closely related to the G\"odel semantics. 
There was the lack of logical characterizations of fuzzy bisimulations between fuzzy graph-based structures over a general residuated lattice, as well as over the residuated lattice $[0,1]$ with the {\L}ukasiewicz or product t-norm. 

In this article, we provide and prove logical characterizations of fuzzy bisimulations in fuzzy modal logics over general residuated lattices. The considered logics are the fuzzy propositional dynamic logic and its fragments. Our logical characterizations concern invariance of formulas under fuzzy bisimulations and the Hennessy-Milner property of fuzzy bisimulations. 
Our results are significant from the theoretical point of view, as they solve \red{the problem stated in the last sentence of the above paragraph}. They would also have an impact on practical applications, e.g.\ for studying logical similarity of individuals and concept learning in fuzzy description logics, as they can be reformulated for other fuzzy structures such as fuzzy interpretations in description logics and fuzzy labeled transition systems, and moreover, residuated lattices cover the lattice $[0,1]$ with any t-norm, including the product and {\L}ukasiewicz t-norms. 

The rest of this article is structured as follows. Section~\ref{section: prel} contains preliminaries for this work. In Section~\ref{section: fbir}, we define fuzzy bisimulations between fuzzy Kripke models and provide some of their basic properties. Sections~\ref{section: ir} and~\ref{sec: HM-prop} contain our results on invariance of formulas under fuzzy bisimulations and the Hennessy-Milner property of fuzzy bisimulations, respectively. Section~\ref{sec: rw} contains a discussion on related work. Conclusions are given in Section~\ref{sec: conc}. The work also contains two appendices: the first one is the proof of a lemma, whereas the second one is a discussion on the relationship with fuzzy bisimulations between fuzzy automata~\cite{CiricIDB12}.


\section{Preliminaries}
\label{section: prel}

In this section, we recall definitions and properties of residuated lattices and fuzzy sets, then present the syntax and semantics of the fuzzy modal logics considered in this article, together with some related notions. 

\subsection{Residuated Lattices and Fuzzy Sets}

A {\em residuated lattice} \cite{Hajek1998,Belohlavek2002} is an algebra $\mL = \tuple{L, \leq, \fand, \fto, 0, 1}$ such that
\begin{itemize}
\item $\tuple{L, \leq, 0, 1}$ is a bounded lattice with the least element 0 and the greatest element 1,
\item $\tuple{L, \fand, 1}$ is a commutative monoid, 
\item $\fand$ and $\fto$ form an adjoint pair, which means that, for every $x, y, z \in L$,
\begin{equation}
x \fand y \leq z \ \ \textrm{iff}\ \ x \leq (y \fto z). \label{fop: GDJSK 00} 
\end{equation}
\end{itemize}

The expression \mbox{$y \fto z$} is called the {\em residual} of $z$ by $y$. Given a residuated lattice \mbox{$\mL = \tuple{L, \leq, \fand, \fto, 0, 1}$}, let $\land$ and $\lor$ denote the {\em join} and {\em meet} operators associated with the lattice. By $x \fequiv y$ we denote \mbox{$(x \fto y) \land (y \fto x)$}. We use the convention that $\fand$ and $\land$ bind stronger than $\lor$, which in turn binds stronger than $\fto$ and $\fequiv$. 

A residuated lattice $\mL = \tuple{L, \leq, \fand, \fto, 0, 1}$ is {\em complete} (resp.\ {\em linear}) if the bounded lattice \mbox{$\tuple{L, \leq, 0, 1}$} is complete (resp.\ linear). It is a {\em Heyting algebra} if $\fand$ is the same as~$\land$. 

We will need the following lemma. Although most assertions of this lemma are well-known~\cite{Belohlavek2002,Hajek1998} and the proof of this lemma is straightforward, we present the proof in~\ref{appendix A} to make this article self-contained.  

\begin{lemma}[cf.~\cite{Belohlavek2002,Hajek1998}]\label{lemma: JHFJW}
Let $\mL = \tuple{L, \leq, \fand, \fto, 0, 1}$ be a residuated lattice.
The following properties hold for all $x,x',y,y',z \in L$:
\begin{eqnarray}
\!\!\!\!\!\!\!\!\!\! x \leq x' \textrm{ and } y \leq y' & \!\!\textrm{implies}\!\! & x \fand y \leq x' \fand y' \label{fop: GDJSK 10}\\
\!\!\!\!\!\!\!\!\!\! x' \leq x \textrm{ and } y \leq y' & \!\!\textrm{implies}\!\! & (x \fto y) \leq (x' \fto y') \label{fop: GDJSK 20}\\
\!\!\!\!\!\!\!\!\!\! x \leq y & \textrm{iff} & (x \fto y) = 1 \label{fop: GDJSK 30} \\
\!\!\!\!\!\!\!\!\!\! x \fand 0 & = & 0 \label{fop: GDJSK 40} \\
\!\!\!\!\!\!\!\!\!\! x \fand (y \lor z) & = & x \fand y \,\lor\, x \fand z \label{fop: GDJSK 50} \\
\!\!\!\!\!\!\!\!\!\! x \fand (x \fto y) & \leq & y \label{fop: GDJSK 60} \\
\!\!\!\!\!\!\!\!\!\! x \fand (y \fto z) & \leq & (x \fto y) \fto z \label{fop: GDJSK 70} \\
\!\!\!\!\!\!\!\!\!\! x \fand (y \fequiv z) & \leq & (x \fto y) \fto z \label{fop: GDJSK 70a} \\
\!\!\!\!\!\!\!\!\!\! x \fand (y \fequiv z) & \leq & y \,\fequiv\, x \fand z \label{fop: GDJSK 80} \\
\!\!\!\!\!\!\!\!\!\! x \fto (y \fto z) & = & y \fto (x \fto z) \label{fop: GDJSK 90} \\
\!\!\!\!\!\!\!\!\!\! x \fto (y \fto z) & \leq & x \fand y \,\fto\, z \label{fop: GDJSK 100} \\
\!\!\!\!\!\!\!\!\!\! x \fto (y \fequiv z) & \leq & x \fand y \,\fto\, z \label{fop: GDJSK 110} \\
\!\!\!\!\!\!\!\!\!\! (x \fto y) \fand (y \fto z) & \leq & x \fto z \label{fop: GDJSK 115} \\
\!\!\!\!\!\!\!\!\!\! (x \fequiv y) \fand (y \fequiv z) & \leq & x \fequiv z \label{fop: GDJSK 115a} \\
\!\!\!\!\!\!\!\!\!\! (x \fequiv x') \land (y \fequiv y') & \leq & x \land y \,\fequiv\, x' \land y' \label{fop: GDJSK 120} \\
\!\!\!\!\!\!\!\!\!\! (x \fequiv x') \land (y \fequiv y') & \leq & x \lor y \,\fequiv\, x' \lor y' \label{fop: GDJSK 130} \\
\!\!\!\!\!\!\!\!\!\! x \fequiv y & \leq & (z \fto x) \fequiv (z \fto y) \label{fop: GDJSK 140} \\
\!\!\!\!\!\!\!\!\!\! x \fequiv y & \leq & (x \fto z) \fequiv (y \fto z). \label{fop: GDJSK 150}
\end{eqnarray}
\comment{
If $\mL$ is complete, then the following property holds for all $x \in L$ and $Y \subseteq L$:
\begin{equation}
x \fand \sup Y = \sup \{x \fand y \mid y \in Y\}. \label{fop: GDJSK 155}
\end{equation}
}
In addition, if $\mL$ is a Heyting algebra, then the following properties hold for all $x,x',y,y',z \in L$:
\begin{eqnarray}
\!\!\!\!\!\!\!\! (x \fequiv x') \land (y \fequiv y') & \leq & (x \fto y) \fequiv (x' \fto y') \label{fop: GDJSK 160}\\
\!\!\!\!\!\!\!\! x \leq (y \fequiv z) & \!\textrm{implies}\! & x \fand y = x \fand z. \label{fop: GDJSK 170}
\end{eqnarray}
\end{lemma}

\begin{example}
Consider the case when $L$ is the unit interval $[0,1]$. 
The most well-known operators $\fand$ are the G\"odel, {\L}ukasiewicz and product t-norms. They are specified below together with their corresponding residua~($\fto$).
\begin{center}
	\begin{tabular}{|c|c|c|c|}
		\hline
		& G\"odel & {\L}ukasiewicz & Product \\
		\hline
		$x \fand y$ & $\min\{x,y\}$ & $\max\{0, x+y-1\}$ & $x \cdot y$ \\
		\hline
		$x \fto y$ 
		& 
		\(
		\left\{
		\!\!\!\begin{array}{ll}
		1 & \textrm{if $x \leq y$} \\ 
		y & \textrm{otherwise}
		\end{array}\!\!\!
		\right.
		\)
		& $\min\{1, 1 - x + y\}$
		& 
		\(
		\left\{
		\!\!\!\begin{array}{ll}
		1 & \textrm{if $x \leq y$} \\ 
		y/x & \textrm{otherwise}
		\end{array}\!\!\!
		\right.
		\)	
		\\ \hline
	\end{tabular}
\end{center}
Note that, for all of these cases of~$\fand$, the considered residuated lattice is linear and complete, and the operator $\fand$ is continuous.
\myend
\end{example}

From now on, let $\mL = \tuple{L, \leq, \fand, \fto, 0, 1}$ be an arbitrary residuated lattice. 

Given a set $X$, a function $f: X \to L$ is called a {\em fuzzy set}, as well as a {\em fuzzy subset} of $X$. 
If $f$ is a fuzzy subset of $X$ and $x \in X$, then $f(x)$ means the fuzzy degree \red{in which} $x$ belongs to the subset. 
For $\{x_1,\ldots,x_n\} \subseteq X$ and $\{a_1,\ldots,a_n\} \subseteq L$, we write $\{x_1:a_1$, \ldots, $x_n:a_n\}$ to denote the fuzzy subset $f$ of $X$ such that $f(x_i) = a_i$ for $1 \leq i \leq n$ and $f(x) = 0$ for $x \in X \setminus \{x_1,\ldots,x_n\}$. 
Given fuzzy subsets $f$ and $g$ of $X$, we write $f \leq g$ to denote that $f(x) \leq g(x)$ for all $x \in X$. If $f \leq g$, then we say that $g$ is greater than or equal to~$f$.

A fuzzy subset of $X \times Y$ is called a {\em fuzzy relation} between $X$ and $Y$. A fuzzy relation between $X$ and itself is called a fuzzy relation on~$X$. 

Given $Z: X \times Y \to L$, the converse \mbox{$Z^- : Y \times X \to L$} of $Z$ is defined by $Z^-(y,x) = Z(x,y)$.

If the underlying residuated lattice is complete, then the {\em composition} of fuzzy relations \mbox{$Z_1: X \times Y \to L$} and \mbox{$Z_2: Y \times Z \to L$}, denoted by $Z_1 \circ Z_2$, is defined to be the fuzzy relation between $X$ and $Z$ such that $(Z_1 \circ Z_2)(x,z) = \sup\{Z_1(x,y) \fand Z_2(y,z) \mid y \in Y\}$ for all $\tuple{x,z} \in X \times Z$. 

Let $\mZ$ be a set of fuzzy relations between $X$ and $Y$. If $\mZ$ is finite or the underlying residuated lattice is complete, then by $\sup\mZ$ we denote the fuzzy relation between $X$ and $Y$ specified by: $(\sup\mZ)(x,y) = \sup\{Z(x,y) \mid Z \in \mZ\}$ for $\tuple{x,y} \in X \times Y$. 
We write $Z_1 \cup Z_2$ to denote $\sup\{Z_1,Z_2\}$. 

A fuzzy relation $Z: X \times X \to L$ is 
\begin{itemize}
\item {\em reflexive} if $Z(x,x) = 1$ for all $x \in X$, 
\item {\em symmetric} if $Z(x,y) = Z(y,x)$ for all $x,y \in X$,
\item {\em transitive} if $Z(x,y) \fand Z(y,z) \leq Z(x,z)$ for all $x,y,z \in X$. 
\end{itemize}
It is a {\em fuzzy equivalence relation} if it is reflexive, symmetric and transitive. 

\subsection{Fuzzy Modal Logics}

Let $\SA$ denote a non-empty set of {\em actions}, which are also called {\em atomic programs}, 
and let $\SP$ denote a non-empty set of {\em propositions}, which are also called {\em atomic formulas}. 
The pair $\tuple{\SA,\SP}$ forms the {\em signature} for the fuzzy modal logics considered in this article. 

Let $\Phi \subseteq \{\cup,\to,?\}$. 
By \fPDLwP we denote the fuzzy propositional dynamic logic without the union program constructor if $\cup$ belongs to $\Phi$, without the test operator if $?$ belongs to $\Phi$, and without the full version of implication if $\to$ belongs to $\Phi$. 

In the following, an expression like $\cup \notin \Phi$ can be read as ``$\cup$ is not excluded''.  

{\em Programs} and {\em formulas} of \fPDLwP over a residuated lattice $\mL = \tuple{L, \leq, \fand, \fto, 0, 1}$ are defined as follows:
\begin{itemize}
\item if $\varrho \in \SA$, then $\varrho$ is a program of \fPDLwP, 
\item if $\alpha$ and $\beta$ are programs of \fPDLwP, then 
	\begin{itemize}
	\item $\alpha \circ \beta$ and $\alpha^*$ are programs of \fPDLwP, 
	\item if $\cup \notin \Phi$, then $\alpha \cup \beta$ is a program of \fPDLwP, 
	\item if $? \notin \Phi$ and $\varphi$ is a formula of \fPDLwP, then $\varphi?$ is a program of \fPDLwP, 
	\end{itemize}
\item if $a \in L$, then $a$ is a formula of \fPDLwP, 
\item if $p \in \SP$, then $p$ is a formula of \fPDLwP, 
\item if $\varphi$ and $\psi$ are formulas of \fPDLwP, $\alpha$ is a program of \fPDLwP and $a \in L$, then $\varphi \land \psi$, $\varphi \lor \psi$, $a \to \varphi$, $\varphi \to a$, $\lnot\varphi$, $[\alpha]\varphi$ and $\tuple{\alpha}\varphi$ are formulas of \fPDLwP,   
\item if $\varphi$ and $\psi$ are formulas of \fPDLwP and $\to\ \notin \Phi$, then $\varphi \to \psi$ is a formula of \fPDLwP.
\end{itemize} 

Note that, even when $\to\ \in \Phi$, \fPDLwP allows implications of the form $a \to \varphi$ or $\varphi \to a$ with $a \in L$. 
By \fPDL we denote \fPDLwP with $\Phi = \emptyset$. 
By \fKz we denote the largest sublanguage of \fPDLwP with $\Phi = \{\cup$, $\to$, $?\}$ that disallows the remaining program constructors ($\alpha \circ \beta$ and $\alpha^*$) and the formula constructors $\lnot\varphi$, $\varphi \lor \psi$ and $[\alpha]\varphi$. 
That is, formulas of \fKz are of the form $a$, $p$, $\varphi \land \psi$, $a \to \varphi$, $\varphi \to a$ or $\tuple{\varrho}\varphi$, where $a \in L$, $p \in \SP$, $\varrho \in \SA$, and $\varphi$ and $\psi$ are formulas of~\fKz.

We use letters like 
\begin{itemize}
\item $\varrho$ to denote actions from $\SA$, 
\item $p$ and $q$ to denote propositions from $\SP$, 
\item $a$ and $b$ to denote values from $L$, 
\item $\varphi$ and $\psi$ to denote formulas, 
\item $\alpha$ and $\beta$ to denote programs.
\end{itemize}

Given a finite set $\Gamma = \{\varphi_1,\ldots,\varphi_n\}$ with \red{$n \geq 0$}, we denote
\begin{eqnarray*}
\textstyle\bigwedge\Gamma & = & \varphi_1 \land\ldots\land \varphi_n \land 1, \\
\textstyle\bigotimes\Gamma & = & \varphi_1 \fand\cdots\fand \varphi_n \fand 1.
\end{eqnarray*}

\begin{definition}\label{def: HFHAD}
A {\em fuzzy Kripke model} over a signature $\tuple{\SA,\SP}$ and a residuated lattice $\mL = \tuple{L, \leq, \fand, \fto, 0, 1}$ is a pair $\mM = \tuple{\Delta^\mM, \cdot^\mM}$, where $\Delta^\mM$ is a non-empty set, called the {\em domain}, and $\cdot^\mM$ is the {\em interpretation function} that maps each $p \in \SP$ to a fuzzy set $p^\mM: \Delta^\mM \to L$ and maps each $\varrho \in \SA$ to a fuzzy relation $\varrho^\mM: \Delta^\mM \times \Delta^\mM \to L$. 
The interpretation function is extended to complex programs and formulas as follows, under the condition that the used suprema and infima exist (e.g., by requiring $\mL$ to be complete or $\mM$ to be witnessed as defined shortly).  
\begin{eqnarray*}
(\varphi?)^\mM(x,y) & \!=\! & \textrm{(if $x = y$ then $\varphi^\mM(x)$ else 0)} \\
(\alpha \cup \beta)^\mM(x,y) & \!=\! & \alpha^\mM(x,y) \lor \beta^\mM(x,y) \\
(\alpha \circ \beta)^\mM(x,y) & \!=\! & \sup\{\alpha^\mM(x,z) \fand \beta^\mM(z,y) \mid z \in \Delta^\mM \} \\
(\alpha^*)^\mM(x,y) & \!=\! & \sup \{\textstyle\bigotimes\{\alpha^\mM(x_i,x_{i+1}) \mid 0 \leq i < n\} \mid \\
	&& \qquad n \geq 0,\ x_0,\ldots,x_n \in \Delta^\mM, x_0 = x,\ x_n = y\} \\
a^\mM(x) & \!=\! & a \\
(\varphi \land \psi)^\mM(x) & \!=\! & \varphi^\mM(x) \land \psi^\mM(x) \\
(\varphi \lor \psi)^\mM(x) & \!=\! & \varphi^\mM(x) \lor \psi^\mM(x) \\
(\varphi \to \psi)^\mM(x) & \!=\! & (\varphi^\mM(x) \fto \psi^\mM(x)) \\
(\lnot\varphi)^\mM(x) & \!=\! & (\varphi \to 0)^\mM(x) \\
([\alpha]\varphi)^\mM(x) & \!=\! & \inf \{\alpha^\mM(x,y) \fto \varphi^\mM(y) \mid y \in \Delta^\mM\} \\
(\tuple{\alpha}\varphi)^\mM(x) & \!=\! & \sup \{\alpha^\mM(x,y) \fand \varphi^\mM(y) \mid y \in \Delta^\mM\}.
\end{eqnarray*}

\vspace{-3.5ex}

\myend
\end{definition}

\begin{example}
Let $\SA = \{\varrho\}$, $\SP = \{p\}$ and let $L$ be the unit interval $[0,1]$. 
Consider the fuzzy Kripke model $\mM$ specified by 
$\Delta^\mM = \{u,v,w\}$, 
$p^\mM = \{u\!:\!0.9,\ v\!:\!0.5,\ w\!:\!0.8\}$, 
$\varrho^\mM = \{\tuple{u,v}\!:\!0.6$, $\tuple{u,w}\!:\!0.7\}$ and depicted below:
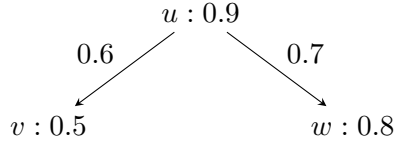
\begin{center}
\begin{tikzpicture}[->,>=stealth]
\node (a) {$u: 0.9$};
\node (ua) [node distance=1.5cm, below of=a] {};
\node (b) [node distance=2cm, left of=ua] {$v: 0.5$};
\node (c) [node distance=2cm, right of=ua] {$w: 0.8$};
\draw (a) to node [left,yshift=2mm]{0.6} (b);
\draw (a) to node [right,yshift=2mm]{0.7} (c);
\end{tikzpicture}
\end{center}
The values $\varphi^\mM(u)$ for some example formulas $\varphi$ using the G\"odel, {\L}ukasiewicz or product t-norm~$\fand$ are given below:
\begin{center}
\begin{tabular}{|c|c|c|c|}
\hline
& \ \ \ \ G\"odel\ \ \ \ & \ \ {\L}ukasiewicz\ \ \ & Product \\
\hline
& & \\[-2.3ex]
$(\tuple{\varrho}p)^\mM(u)$ & $0.7$ & $0.5$ & $0.56$ \\
\hline
& & \\[-2.3ex]
$([\varrho]p)^\mM(u)$ & $0.5$ & $0.9$ & $5/6$ \\
\hline
& & \\[-2.3ex]
$(\tuple{\varrho^*}p)^\mM(u)$ & $0.9$ & $0.9$ & $0.9$ \\
\hline
& & \\[-2.3ex]
$([\varrho^*]p)^\mM(u)$ & $0.5$ & $0.9$ & $5/6$ \\
\hline
\end{tabular}
\end{center}

\vspace{-2.5ex}

\myend
\end{example}

A fuzzy Kripke model $\mM$ is {\em witnessed} w.r.t.\ \fPDLwP if every infinite set under the infimum (resp.\ supremum) operator in Definition~\ref{def: HFHAD} has a smallest (resp.\ biggest) element when considering only formulas and programs of \fPDLwP (cf.~\cite{Hajek05}). 
The notion of whether a fuzzy Kripke model $\mM$ is {\em witnessed} w.r.t.\ \fKz is defined analogously by restricting to formulas and programs of \fKz. 

A fuzzy Kripke model $\mM$ is {\em image-finite} if, for every $x \in \Delta^\mM$ and every $\varrho \in \SA$, the set $\{y \in \Delta^\mM \mid \varrho^\mM(x,y) > 0\}$ is finite. It is {\em finite} if \red{$\Delta^\mM$ is finite}.

Observe that every finite fuzzy Kripke model is witnessed w.r.t.\ \fPDLwP (and hence also w.r.t.\ \fKz) and every image-finite fuzzy Kripke model is witnessed w.r.t.~\fKz. If the underlying residuated lattice is finite, then all fuzzy Kripke models are witnessed w.r.t.\ \fPDLwP. 


\section{Fuzzy Bisimulations between Kripke Models}
\label{section: fbir}

In this section, we define fuzzy bisimulations between fuzzy Kripke models, then state and prove some of their basic properties. 
The relationship with fuzzy bisimulations between fuzzy automata~\cite{CiricIDB12} is presented in~\ref{appendix B}. 

\begin{definition}\label{def: HDJAK}
Given fuzzy Kripke models $\mM$ and $\mMp$, a fuzzy relation \mbox{$Z : \Delta^\mM \times \Delta^\mMp \to L$} is called a {\em fuzzy bisimulation} between $\mM$ and $\mMp$ if the following conditions hold for all $p \in \SP$, $\varrho \in \SA$ and all possible values for the free variables:
\begin{eqnarray}
\!\!\!\!\!\!\!\!\!\!&& Z(x,x') \leq (p^\mM(x) \fequiv p^\mMp(\red{x'})) \label{eq: FB1} \\
\!\!\!\!\!\!\!\!\!\!&& \E y' \in \Delta^\mMp\ (Z(x,x') \fand \varrho^\mM(x,y) \leq \varrho^\mMp(x',y') \fand Z(y,y')) \label{eq: FB2} \\
\!\!\!\!\!\!\!\!\!\!&& \E y \in \Delta^\mM\ (Z(x,x') \fand \varrho^\mMp(x',y') \leq \varrho^\mM(x,y) \fand Z(y,y')). \label{eq: FB3}
\end{eqnarray}
\end{definition}

\begin{example}
	Let $\SA = \{\varrho\}$, $\SP = \{p\}$ and $L = [0,1]$. Consider the fuzzy Kripke models $\mM$ and $\mMp$ depicted and specified below. 
	
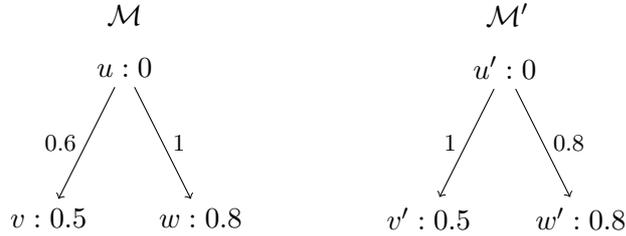
\begin{center}
		\begin{tikzpicture}
		\node (x0) {};
		\node (x) [node distance=3.0cm, right of=x0] {};
		\node (I) [node distance=0.0cm, below of=x] {$\mM$};
		\node (I1) [node distance=5.0cm, right of=I] {$\mMp$};
		\node (u) [node distance=0.7cm, below of=I] {$u:0$};
		\node (ub) [node distance=2.0cm, below of=u] {};
		\node (v) [node distance=1.0cm, left of=ub] {$v:0.5$};
		\node (w) [node distance=1.0cm, right of=ub] {$w:0.8$};
		\node (up) [node distance=0.7cm, below of=I1] {$u':0$};
		\node (ubp) [node distance=2.0cm, below of=up] {};
		\node (vp) [node distance=1.0cm, left of=ubp] {$v':0.5$};
		\node (wp) [node distance=1.0cm, right of=ubp] {$w':0.8$};
		\draw[->] (u) to node [left]{\footnotesize{0.6}} (v);
		\draw[->] (u) to node [right]{\footnotesize{1}} (w);
		\draw[->] (up) to node [left]{\footnotesize{1}} (vp);
		\draw[->] (up) to node [right]{\footnotesize{0.8}} (wp);
		\end{tikzpicture}
	\end{center}
	\begin{itemize}
		\item $\Delta^\mM = \{u,v,w\}$,\ \ $\Delta^\mMp = \{u',v',w'\}$, 
		\item $p^\mM = \{u\!:\!0,\, v\!:\!0.5,\, w\!:\!0.8\}$,\ \ $p^\mMp = \{u'\!:\!0,\, v'\!:\!0.5,\, w'\!:\!0.8\}$,
		\item $\varrho^\mM = \{\tuple{u,v}\!:\!0.6, \tuple{u,w}\!:\!1\}$,\ \ $\varrho^\mMp = \{\tuple{u',v'}\!:\!1, \tuple{u',w'}\!:\!0.8\}$.
	\end{itemize}
	
	In the case when $\fand$ is the G\"odel, {\L}ukasiewicz or product t-norm, the greatest fuzzy bisimulation $Z$ between $\mM$ and $\mMp$ can be computed as follows:
	\begin{itemize}
		\item $Z(v,v') = (0.5 \fequiv 0.5) = 1$,\ \ $Z(w,w') = (0.8 \fequiv 0.8) = 1$;
		\item $Z(v,w') = (0.5 \fequiv 0.8)$,\ \ $Z(w,v') = (0.8 \fequiv 0.5)$;
		\item $Z(v,u') \leq (0.5 \fequiv 0) = 0$,\ \ $Z(w,u') \leq (0.8 \fequiv 0) = 0$;
		\item $Z(u,v') \leq (0 \fequiv 0.5) = 0$,\ \ $Z(u,w') \leq (0\fequiv 0.8) = 0$;
		\item $(\varrho^\mMp \circ Z^-)(u',v) = 1$,\ \ $(\varrho^\mMp \circ Z^-)(u',w) = \max\{0.8, (0.8 \fequiv 0.5)\} = 0.8$,\\
		thus, the condition~\eqref{eq: ERJSJ 2} only requires $Z(u,u') \leq 0.8$;
		\item $(\varrho^\mM \circ Z)(u,w') = 1$,\ \ $(\varrho^\mM \circ Z)(u,v') = \max\{0.6, (0.8 \fequiv 0.5)\}$,\\
		thus, the condition~\eqref{eq: ERJSJ 3} only requires $Z(u,u') \leq \max\{0.6, (0.8 \fequiv 0.5)\}$;
		\item therefore, $Z(u,u') = \max\{0.6, (0.8 \fequiv 0.5)\}$.
	\end{itemize}
	That is, 
	\begin{itemize}
		\item if $\fand$ is the G\"odel t-norm, then 
		\[ Z = \{\tuple{u,u'}\!:\!0.6,\ \tuple{v,v'}\!:\!1,\ \tuple{w,w'}\!:\!1,\ \tuple{v,w'}\!:\!0.5,\ \tuple{w,v'}\!:\!0.5\}; \]
		\item if $\fand$ is the {\L}ukasiewicz t-norm, then 
		\[ Z = \{\tuple{u,u'}\!:\!0.7,\ \tuple{v,v'}\!:\!1,\ \tuple{w,w'}\!:\!1,\ \tuple{v,w'}\!:\!0.7,\ \tuple{w,v'}\!:\!0.7\}; \]
		\item if $\fand$ is the product t-norm, then 
		\[ \qquad\qquad\quad Z = \{\tuple{u,u'}\!:\!0.625,\ \tuple{v,v'}\!:\!1,\ \tuple{w,w'}\!:\!1,\ \tuple{v,w'}\!:\!0.625,\ \tuple{w,v'}\!:\!0.625\}.\qquad\qquad\quad\Myend \]
	\end{itemize}
\end{example}

\begin{proposition}\label{prop: HFKSS}
	Suppose that the underlying residuated lattice is complete. 
	If $Z$ is a fuzzy bisimulation between fuzzy Kripke models $\mM$ and $\mMp$, then it satisfies the following conditions for all $x \in \Delta^\mM$, $x' \in \Delta^\mMp$ and $\varrho \in \SA$:
	\begin{eqnarray}
	Z(x,x') & \leq & \inf\{p^\mM(x) \fequiv p^\mMp(\red{x'}) \mid p \in \SP\} \label{eq: ERJSJ 1}\\
	Z^- \circ \varrho^\mM & \leq & \varrho^\mMp \circ Z^- \label{eq: ERJSJ 2}\\
	Z \circ \varrho^\mMp & \leq & \varrho^\mM \circ Z. \label{eq: ERJSJ 3}
	\end{eqnarray}
	Conversely, if the underlying residuated lattice is also linear, fuzzy Kripke models $\mM$ and $\mMp$ are image-finite and \mbox{$Z: \Delta^\mM \times \Delta^\mMp \to L$} is a fuzzy relation satisfying the conditions~\eqref{eq: ERJSJ 1}--\eqref{eq: ERJSJ 3}, then $Z$ is a fuzzy bisimulation between~$\mM$ and~$\mMp$. 
\end{proposition}

\begin{proof}\markRed
Suppose that $Z$ is a fuzzy bisimulation between fuzzy Kripke models $\mM$ and $\mMp$. Let $x \in \Delta^\mM$, $x' \in \Delta^\mMp$ and $\varrho \in \SA$. We show that $Z$ satisfies the assertions \eqref{eq: ERJSJ 1}--\eqref{eq: ERJSJ 3}. 

The assertion~\eqref{eq: ERJSJ 1} holds because \eqref{eq: FB1} holds for all $p \in \SP$. 

Consider the assertion~\eqref{eq: ERJSJ 2}. We need to prove that, for every $x' \in \Delta^\mMp$ and $y \in \Delta^\mM$, 
\[ (Z^- \circ \varrho^\mM)(x',y) \leq (\varrho^\mMp \circ Z^-)(x',y). \]
It is sufficient to show that, for every $x \in \Delta^\mM$, 
\[ Z(x,x') \fand \varrho^\mM(x,y) \leq (\varrho^\mMp \circ Z^-)(x',y). \]
This inequality follows from the condition~\eqref{eq: FB2}. 

Consider the assertion~\eqref{eq: ERJSJ 3}. We need to prove that, for every $x \in \Delta^\mM$ and $y' \in \Delta^\mMp$, 
\[ (Z \circ \varrho^\mMp)(x,y') \leq (\varrho^\mM \circ Z)(x,y'). \]
It is sufficient to show that, for every $x' \in \Delta^\mMp$, 
\[ Z(x,x') \fand \varrho^\mMp(x',y') \leq (\varrho^\mM \circ Z)(x,y'). \]
This inequality follows from the condition~\eqref{eq: FB3}.

For the converse, suppose that the underlying residuated lattice is linear, $\mM$ and $\mMp$ are image-finite fuzzy Kripke models and \mbox{$Z: \Delta^\mM \times \Delta^\mMp \to L$} is a fuzzy relation satisfying the conditions~\eqref{eq: ERJSJ 1}--\eqref{eq: ERJSJ 3}. We prove that $Z$ is a fuzzy bisimulation between~$\mM$ and~$\mMp$. Let $p \in \SP$, $\varrho \in \SA$, $x \in \Delta^\mM$ and $x' \in \Delta^\mMp$. We need to show that $Z$ satisfies the assertions~\eqref{eq: FB1}--\eqref{eq: FB3}, for any $y \in \Delta^\mM$ or $y' \in \Delta^\mMp$ if it is a free variable.

The assertion~\eqref{eq: FB1} follows from~\eqref{eq: ERJSJ 1}. 

Consider the assertion~\eqref{eq: FB2} for any $y \in \Delta^\mM$. By definition, 
\[ 
	Z(x,x') \fand \varrho^\mM(x,y) \leq (Z^- \circ \varrho^\mM)(x',y).
\]
By~\eqref{eq: ERJSJ 2}, it follows that  
\[ 
	Z(x,x') \fand \varrho^\mM(x,y) \leq (\varrho^\mMp \circ Z^-)(x',y).
\]
Since the underlying residuated lattice is linear and $\mMp$ is image-finite, there exists $y' \in \Delta^\mMp$ such that 
\[ 
	(\varrho^\mMp \circ Z^-)(x',y) = \varrho^\mMp(x',y') \fand Z(y,y'). 
\]
Therefore, the following inequality holds, which implies~\eqref{eq: FB2}: 
\[ 
	Z(x,x') \fand \varrho^\mM(x,y) \leq \varrho^\mMp(x',y') \fand Z(y,y').
\]

Consider the assertion~\eqref{eq: FB3} for any $y' \in \Delta^\mMp$. By definition, 
\[
	Z(x,x') \fand \varrho^\mMp(x',y') \leq (Z \circ \varrho^\mMp)(x,y').
\]
By~\eqref{eq: ERJSJ 3}, it follows that 
\[
	Z(x,x') \fand \varrho^\mMp(x',y') \leq (\varrho^\mM \circ Z)(x,y').
\]
Since the underlying residuated lattice is linear and $\mM$ is image-finite, there exists $y \in \Delta^\mM$ such that 
\[
	(\varrho^\mM \circ Z)(x,y') = \varrho^\mM(x,y) \fand Z(y,y'). 
\]
Therefore, the following inequality holds, which implies~\eqref{eq: FB3}: 
\[
	Z(x,x') \fand \varrho^\mMp(x',y') \leq \varrho^\mM(x,y) \fand Z(y,y').
\]

\vspace{-2ex}

\myend
\end{proof}

\begin{example}\markRed
Let the underlying residuated lattice use $L = \{0,a,b,1\}$ (with four pairwise distinct elements) where $a$ and $b$ are not comparable. Let $\SA = \{\varrho\}$ and $\SP = \emptyset$. Consider the Kripke models $\mM$ and $\mMp$ specified and depicted below:
\begin{itemize}
\item $\Delta^\mM = \{u,v\}$,\ \ $\varrho^\mM = \{\tuple{u,v}:1\}$,
\item $\Delta^\mMp = \{u',v'_1, v'_2\}$,\ \ $\varrho^\mMp = \{\tuple{u',v'_1}\!:\!a, \tuple{u',v'_2}\!:\!b\}$.
\end{itemize}
\begin{center}
	\begin{tikzpicture}
	\node (I) {$\mM$};
	\node (I1) [node distance=6.0cm, right of=I] {$\mMp$};
	\node (u) [node distance=0.7cm, below of=I] {$u$};
	\node (v) [node distance=1.5cm, below of=u] {$v$};
	\node (up) [node distance=0.7cm, below of=I1] {$u'$};
	\node (ubp) [node distance=1.5cm, below of=up] {};
	\node (v1) [node distance=1.0cm, left of=ubp] {$v'_1$};
	\node (v2) [node distance=1.0cm, right of=ubp] {$v'_2$};
	\draw[->] (u) to node [left]{\footnotesize{1}} (v);
	\draw[->] (up) to node [left]{a} (v1);
	\draw[->] (up) to node [right]{b} (v2);
	\end{tikzpicture}
\end{center}
Let $Z: \Delta^\mM \times \Delta^\mMp \to L$ be the fuzzy relation specified by $Z = \{\tuple{u,u'}\!:\!1, \tuple{v,v'_1}\!:\!1, \tuple{v,v'_2}\!:\!1\}$. 
Observe that $Z$ satisfies the conditions~\eqref{eq: ERJSJ 1}--\eqref{eq: ERJSJ 3}, but it is not a fuzzy bisimulation between $\mM$ and $\mMp$, as it does not satisfy the condition~\eqref{eq: FB2}. The reason is that the underlying residuated lattice is not linear.
\myend
\end{example}

\begin{example}\label{example: HGDHJ}\markRed{}
Let $L$ be the unit interval $[0,1]$ and $\fand$ the G\"odel t-norm. Let $\SA = \{\varrho\}$ and $\SP = \emptyset$. Consider the Kripke models $\mM$ and $\mMp$ specified and depicted below (cf.~\cite{CaoCK11,FSS2020}):
\begin{itemize}
\item $\Delta^\mM = \{u,v\}$,\ \ $\varrho^\mM = \{\tuple{u,v}:1\}$,
\item $\Delta^\mMp = \{u',v'_i \mid i \in \NN \setminus \{0\}\}$,\ \ $\varrho^\mMp = \{\tuple{u',v'_i}\!:\!\frac{i}{i+1} \mid i \in \NN \setminus \{0\}\}$.
\end{itemize}
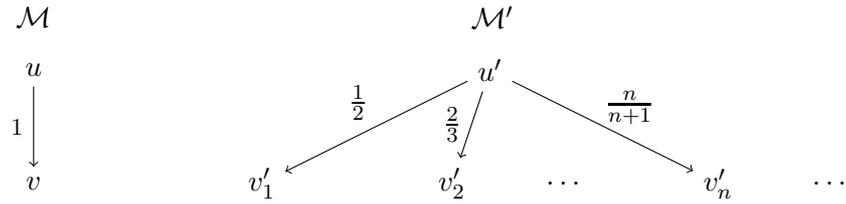
\begin{center}
	\begin{tikzpicture}
	\node (x0) {};
	\node (x) [node distance=3.0cm, right of=x0] {};
	\node (I) [node distance=0.0cm, below of=x] {$\mM$};
	\node (I1) [node distance=6.0cm, right of=I] {$\mMp$};
	\node (u) [node distance=0.7cm, below of=I] {$u$};
	\node (v) [node distance=1.5cm, below of=u] {$v$};
	\node (up) [node distance=0.7cm, below of=I1] {$u'$};
	\node (ubp) [node distance=1.5cm, below of=up] {};
	\node (v1) [node distance=3.0cm, left of=ubp] {$v'_1$};
	\node (v2) [node distance=0.5cm, left of=ubp] {$v'_2$};
	\node (v3) [node distance=1.0cm, right of=ubp] {\ldots};
	\node (vn) [node distance=3.0cm, right of=ubp] {$v'_n$};
	\node (vnp) [node distance=4.5cm, right of=ubp] {\ldots};
	\draw[->] (u) to node [left]{\footnotesize{1}} (v);
	\draw[->] (up) to node [above=2pt, pos=0.60]{$\frac{1}{2}$} (v1);
	\draw[->] (up) to node [left]{$\frac{2}{3}$} (v2);
	\draw[->] (up) to node [above=2pt, pos=0.65]{$\frac{n}{n+1}$} (vn);
	\end{tikzpicture}
\end{center}
Let $Z: \Delta^\mM \times \Delta^\mMp \to L$ be the fuzzy relation specified by $Z = \{\tuple{u,u'}\!:\!1, \tuple{v,v'_i}\!:\!1 \mid i \in \NN \setminus \{0\}\}$. 
Observe that $Z$ satisfies the conditions~\eqref{eq: ERJSJ 1}--\eqref{eq: ERJSJ 3}, but it is not a fuzzy bisimulation between $\mM$ and $\mMp$, as it does not satisfy the condition~\eqref{eq: FB2}. The reason is that $\mMp$ is not image-finite. 
\myend
\end{example}

The above proposition is related to Remark~3.4 of~\cite{FSS2020}. 
In~\cite{Fan15}, Fan studied fuzzy bisimulations that are defined for fuzzy Kripke models (over a signature with $|\SA| = 1$) using conditions like~\eqref{eq: ERJSJ 1}--\eqref{eq: ERJSJ 3} and the G\"odel semantics over the unit interval $[0,1]$. In~\cite{FSS2020}, Nguyen {\em et al.}~studied fuzzy bisimulations that are defined for interpretations in description logics using conditions like~\eqref{eq: FB1}--\eqref{eq: FB3} and the G\"odel semantics over the unit interval $[0,1]$. Note that the residuated lattices used in both the works~\cite{Fan15,FSS2020} are linear, complete and use a fixed operator $\fand$, which is the G\"odel t-norm. 
The relationship between~\eqref{eq: FB1}--\eqref{eq: FB3} and~\eqref{eq: ERJSJ 1}--\eqref{eq: ERJSJ 3} is characterized by the above proposition. On one hand, the conditions~\eqref{eq: FB1}--\eqref{eq: FB3} do not require the underlying residuated lattice to be complete and, as discussed in~\cite{FSS2020}, the style is appropriate for the extension that deals with number restrictions (in description logics) and graded modalities. On the other hand, when restricting to complete residuated lattices and the case without graded modalities, the conditions~\eqref{eq: ERJSJ 1}--\eqref{eq: ERJSJ 3} are weaker\footnote{\red{\eqref{eq: FB2} implies \eqref{eq: ERJSJ 2}, but not vice versa; similarly, \eqref{eq: FB3} implies \eqref{eq: ERJSJ 3}, but not vice versa.}} and, \red{when used instead of \eqref{eq: FB1}--\eqref{eq: FB3}}, make the notion of fuzzy bisimulation stronger\footnote{\red{in the sense that more fuzzy relations can be fuzzy bisimulations (see, e.g., Example~\ref{example: HGDHJ}).}} for non-image-finite fuzzy Kripke models. 

A fuzzy bisimulation between $\mM$ and itself is called a {\em fuzzy auto-bisimulation} of~$\mM$.

\begin{proposition}\label{prop: HFHSJ}
Let $\mM$, $\mMp$ and $\mM''$ be \red{image-finite} fuzzy Kripke models.
\begin{enumerate}
\item\label{ass: HFHSJ 1} \red{The fuzzy relation $Z : \Delta^\mM \times \Delta^\mM \to L$ specified by $Z(x,x')$ = (if} $x = x'$ then 1 else 0) is a fuzzy auto-bisimulation of~$\mM$.
\item\label{ass: HFHSJ 2} If $Z$ is a fuzzy bisimulation between $\mM$ and $\mMp$, then $Z^-$ is a fuzzy bisimulation between $\mMp$ and~$\mM$.
\item\label{ass: HFHSJ 3} If the underlying residuated lattice is \red{linear and} complete, $Z_1$ is a fuzzy bisimulation between $\mM$ and $\mMp$, and $Z_2$ is a fuzzy bisimulation between $\mMp$ and $\mM''$, then $Z_1 \circ Z_2$ is a fuzzy bisimulation between $\mM$ and $\mM''$.
\item\label{ass: HFHSJ 4} If the underlying residuated lattice is linear and $\mZ$ is a finite set of fuzzy bisimulations between $\mM$ and $\mMp$, then $\sup\mZ$ is also a fuzzy bisimulation between $\mM$ and $\mMp$.
\end{enumerate}   
\end{proposition}

\begin{proof}
The proofs of the first two assertions are straightforward.

Consider the third assertion and assume that the premises hold. We have to show that $Z_1 \circ Z_2$ satisfies the conditions \eqref{eq: FB1}--\eqref{eq: FB3}. 
\begin{itemize}
\item Consider the condition~\eqref{eq: FB1}. 
Let $x \in \Delta^\mM$, $x' \in \Delta^\mMp$, $x'' \in \Delta^{\mM''}$ and $p \in \SP$. 
We have that 
\begin{eqnarray*}
Z_1(x,x') & \leq & p^\mM(x) \fequiv p^\mMp(\red{x'}) \\
Z_2(x',x'') & \leq & p^\mMp(x') \fequiv p^{\mM''}(x'').
\end{eqnarray*}
Due to~\eqref{fop: GDJSK 115a}, 
\( (p^\mM(x) \fequiv p^\mMp(\red{x'})) \fand (p^\mMp(x') \fequiv p^{\mM''}(x'')) \leq (p^\mM(x) \fequiv p^{\mM''}(x'')). \)
By~\eqref{fop: GDJSK 10}, it follows that 
\[ Z_1(x,x') \fand Z_2(x',x'') \leq (p^\mM(x) \fequiv p^{\mM''}(x'')). \]
Therefore, $(Z_1 \circ Z_2)(x,x'') \leq (p^\mM(x) \fequiv p^{\mM''}(x''))$, which completes the proof of~\eqref{eq: FB1}.

\item Consider the condition~\eqref{eq: FB2}. 
Let $x,y \in \Delta^\mM$, $x'' \in \Delta^{\mM''}$ and $\varrho \in \SA$. We need to show that there exists $y'' \in \Delta^{\mM''}$ such that 
\begin{equation}
(Z_1 \circ Z_2)(x,x'') \fand \varrho^\mM(x,y) \leq \varrho^{\mM''}(x'',y'') \fand (Z_1 \circ Z_2)(y,y''). \label{eq: HJEAJ 1}
\end{equation}
Let $x'$ be an arbitrary element of $\Delta^\mMp$. Since $Z_1$ is a fuzzy bisimulation between $\mM$ and $\mMp$, there exists $y' \in \mMp$ such that 
\begin{equation}
Z_1(x,x') \fand \varrho^\mM(x,y) \leq \varrho^\mMp(x',y') \fand Z_1(y,y'). \label{eq: HJEAJ 2}
\end{equation}
Since $Z_2$ is a fuzzy bisimulation between $\mMp$ and $\mM''$, there exists $y'' \in \mM''$ such that 
\begin{equation}
Z_2(x',x'') \fand \varrho^\mMp(x',y') \leq \varrho^{\mM''}(x'',y'') \fand Z_2(y',y''). \label{eq: HJEAJ 3}
\end{equation}
By \eqref{eq: HJEAJ 2}, \eqref{eq: HJEAJ 3} and~\eqref{fop: GDJSK 10}, we have that  
\begin{eqnarray*}
Z_2(x',x'') \fand Z_1(x,x') \fand \varrho^\mM(x,y) & \leq & Z_2(x',x'') \fand \varrho^\mMp(x',y') \fand Z_1(y,y') \\
& \leq & \varrho^{\mM''}(x'',y'') \fand Z_2(y',y'') \fand Z_1(y,y') \\
& \leq & \varrho^{\mM''}(x'',y'') \fand (Z_1 \circ Z_2)(y,y''),
\end{eqnarray*}
which implies~\eqref{eq: HJEAJ 1} \red{because $x'$ is an arbitrary element of $\Delta^\mMp$, $\mM''$ is image-finite and the underlying residuated lattice is linear}.

\item The condition~\eqref{eq: FB3} can be proved analogously.
\end{itemize}

Consider now the fourth assertion and assume that the premises hold. It is sufficient to consider the case when $\mZ = \{Z_1,Z_2\}$. We need to prove that $Z_1 \cup Z_2$ satisfies the conditions \eqref{eq: FB1}--\eqref{eq: FB3}. 

\begin{itemize}
\item Consider the condition~\eqref{eq: FB1}. 
Let $x \in \Delta^\mM$, $x' \in \Delta^\mMp$ and $p \in \SP$. 
Since $Z_1$ and $Z_2$ are fuzzy bisimulations between $\mM$ and $\mMp$, we have that 
$Z_1(x,x') \leq (p^\mM(x) \fequiv p^\mMp(\red{x'}))$ and $Z_2(x,x') \leq (p^\mM(x) \fequiv p^\mMp(\red{x'}))$. 
Hence, $(Z_1 \cup Z_2)(x,x') \leq (p^\mM(x) \fequiv p^\mMp(\red{x'}))$. 

\item Consider the condition~\eqref{eq: FB2}. 
Let $x,y \in \Delta^\mM$, $x' \in \Delta^\mMp$ and $\varrho \in \SA$. We need to show that there exists $y' \in \Delta^\mMp$ such that 
\begin{equation}
(Z_1 \cup Z_2)(x,x') \fand \varrho^\mM(x,y) \leq \varrho^\mMp(x',y') \fand (Z_1 \cup Z_2)(y,y'). \label{eq: HJEAJ 4}
\end{equation}
Without loss of generality, assume that $Z_1(x,x') \leq Z_2(x,x')$. Since $Z_2$ is a fuzzy bisimulation between $\mM$ and $\mMp$, there exists $y' \in \Delta^\mMp$ such that 
\[ Z_2(x,x') \fand \varrho^\mM(x,y) \leq \varrho^\mMp(x',y') \fand Z_2(y,y'). \] 
Thus, 
\begin{eqnarray*}
(Z_1 \cup Z_2)(x,x') \fand \varrho^\mM(x,y) & = & Z_2(x,x') \fand \varrho^\mM(x,y) \\
& \leq & \varrho^\mMp(x',y') \fand Z_2(y,y') \\
& \leq & \varrho^\mMp(x',y') \fand (Z_1 \cup Z_2)(y,y').
\end{eqnarray*}

\item The condition~\eqref{eq: FB3} can be proved analogously.
\myend
\end{itemize}
\end{proof}

\begin{corollary}\label{cor: HFHSJ}
Let $\mM$ be \red{an image-finite} fuzzy Kripke model. 
If $Z$ is the greatest fuzzy auto-bisimulation of $\mM$ and the underlying residuated lattice is \red{linear and} complete, then $Z$ is a fuzzy equivalence relation.
\end{corollary}

\begin{proof}
Suppose that $Z$ is the greatest fuzzy auto-bisimulation of $\mM$ and the underlying residuated lattice is complete. 
By the assertion~\ref{ass: HFHSJ 1} of Proposition~\ref{prop: HFHSJ}, $Z$ is reflexive. By the assertion~\ref{ass: HFHSJ 2} of Proposition~\ref{prop: HFHSJ}, $Z^-$ is a fuzzy auto-bisimulation of $\mM$. Hence, $Z^- \leq Z$ and $Z$ is symmetric. By the assertion~\ref{ass: HFHSJ 3} of Proposition~\ref{prop: HFHSJ}, $Z \circ Z$ is a fuzzy auto-bisimulation of $\mM$. Hence, $Z \circ Z \leq Z$ and $Z$ is transitive. Therefore, $Z$ is a fuzzy equivalence relation. 
\myend
\end{proof}


\section{Invariance Results}
\label{section: ir}

We say that a formula $\varphi$ is {\em invariant} under fuzzy bisimulations w.r.t.\ \fPDLwP if, 
for every fuzzy Kripke models $\mM$ and $\mM'$ that are witnessed w.r.t.\ \fPDLwP and for every fuzzy bisimulation $Z$ between $\mM$ and $ \mM'$, \mbox{$Z(x,x') \leq (\varphi^\mM(x) \fequiv \varphi^\mMp(x'))$} for all $x \in \Delta^\mM$ and $x' \in \Delta^\mMp$. 
%
 
\begin{theorem} \label{theorem: invariance}
All formulas of \fPDLwP are invariant under fuzzy bisimulations w.r.t.\ \fPDLwP if the underlying residuated lattice $\mL$ satisfies the following conditions: 
\begin{eqnarray}
&& \textrm{if $\cup \notin \Phi$, then $\mL$ is linear;} \label{req: HFJSK 1} \\
&& \textrm{if $\to\ \notin \Phi$ or $? \notin \Phi$, then $\mL$ is a Heyting algebra.} \label{req: HFJSK 2}
\end{eqnarray}
\end{theorem}

This theorem is an immediate consequence of the following lemma.

\begin{lemma} \label{lemma: GDHAW}
Let $\mM$ and $\mM'$ be fuzzy Kripke models that are witnessed w.r.t.\ \fPDLwP and $Z$ a~fuzzy bisimulation between $\mM$ and $\mM'$. Suppose that the underlying residuated lattice satisfies the conditions~\eqref{req: HFJSK 1} and~\eqref{req: HFJSK 2}. Then, the following properties hold for every formula $\varphi$ of \fPDLwP, every program $\alpha$ of \fPDLwP and every possible values of the free variables:
\begin{eqnarray}
\!\!\!\!\!\!\!\!\!\!&& Z(x,x') \leq (\varphi^\mM(x) \fequiv \varphi^\mMp(x')) \label{eq: GDHAW 1} \\[0.5ex]
\!\!\!\!\!\!\!\!\!\!&& \E y' \in \Delta^\mMp\ (Z(x,x') \fand \alpha^\mM(x,y) \leq \alpha^\mMp(x',y') \fand Z(y,y')) \label{eq: GDHAW 2} \\
\!\!\!\!\!\!\!\!\!\!&& \E y \in \Delta^\mM\ (Z(x,x') \fand \alpha^\mMp(x',y') \leq \alpha^\mM(x,y) \fand Z(y,y')). \label{eq: GDHAW 3} 
\end{eqnarray}
\end{lemma}

\begin{proof}
	We prove this lemma by induction on the structures of $\varphi$ and~$\alpha$.
	First, consider the assertion~\eqref{eq: GDHAW 2}. Let $x,y \in \Delta^\mM$ and $x' \in \Delta^\mMp$. It suffices to show that there exists $y' \in \Delta^\mMp$ such that
	\begin{equation}\label{eq: HDHAK}
	Z(x,x') \fand \alpha^\mM(x,y) \leq \alpha^\mMp(x',y') \fand Z(y,y').
	\end{equation}
	The base case occurs when $\alpha$ is an atomic program and follows from~\eqref{eq: FB2}. The induction steps are given below.
	\begin{itemize}
		\item Case $\alpha = (\psi?)$ (and $? \notin \Phi$): If $x \neq y$, then $\alpha^\mM(x,y) = 0$ and, by~\eqref{fop: GDJSK 40}, the assertion~\eqref{eq: HDHAK} clearly holds. Suppose $x = y$ and take $y' = x'$. By the induction assumption \red{about}~\eqref{eq: GDHAW 1}, \mbox{$Z(x,x') \leq (\psi^\mM(x) \fequiv \psi^\mMp(x'))$}. Hence, by~\eqref{fop: GDJSK 170},  
		\[ Z(x,x') \fand \psi^\mM(x) = Z(x,x') \fand \psi^\mMp(x'), \] 
		which implies \eqref{eq: HDHAK}.
		
		\item Case $\alpha = \beta \cup \gamma$ (and $\cup \notin \Phi$): Without loss of generality, suppose $\beta^\mM(x,y) \geq \gamma^\mM(x,y)$. Thus, $\alpha^\mM(x,y) = \beta^\mM(x,y)$. By the induction assumption of~\eqref{eq: GDHAW 2}, there exists $y' \in \Delta^\mMp$ such that 
		\[ Z(x,x') \fand \beta^\mM(x,y) \leq \beta^\mMp(x',y') \fand Z(y,y'). \] 
		Thus, 
		\begin{eqnarray*}
		Z(x,x') \fand \alpha^\mM(x,y) & = & Z(x,x') \fand \beta^\mM(x,y)\\
		& \leq & \beta^\mMp(x',y') \fand Z(y,y')\\
		& \leq & \alpha^\mMp(x',y') \fand Z(y,y').
		\end{eqnarray*}

		\item Case $\alpha = \beta \circ \gamma$: Since $\mM$ is witnessed w.r.t.\ \fPDLwP, there exists $z \in \Delta^\mM$ such that \mbox{$\alpha^\mM(x,y) = \beta^\mM(x,z) \fand \gamma^\mM(z,y)$}. By the induction assumption of~\eqref{eq: GDHAW 2}, there exist $z'$ and $y'$ such that:
		\begin{eqnarray*}
			Z(x,x') \fand \beta^\mM(x,z) & \leq & \beta^\mMp(x',z') \fand Z(z,z') \\
			Z(z,z') \fand \gamma^\mM(z,y) & \leq & \gamma^\mMp(z',y') \fand Z(y,y').
		\end{eqnarray*}
		Since $\fand$ is associative \red{and due to~\eqref{fop: GDJSK 10}}, it follows that 
		\begin{eqnarray*}
			Z(x,x') \fand \alpha^\mM(x,y) 
			& = & Z(x,x') \fand \beta^\mM(x,z) \fand \gamma^\mM(z,y) \\
			& \leq & \beta^\mMp(x',z') \fand Z(z,z') \fand \gamma^\mM(z,y) \\
			& \leq & \beta^\mMp(x',z') \fand \gamma^\mMp(z',y') \fand Z(y,y') \\
			& \leq & \alpha^\mMp(x',y') \fand Z(y,y').
		\end{eqnarray*}
		
		\item Case $\alpha = \beta^*$: Since $\mM$ is witnessed w.r.t.\ \fPDLwP, there exist $x_0, \ldots, x_k \in \Delta^\mM$ such that $x_0 = x$, $x_k = y$ and 
		\[ \alpha^\mM(x,y) = \beta^\mM(x_0,x_1) \fand\cdots\fand \beta^\mM(x_{k-1},x_k). \]
		Let $x'_0 = x'$. By the induction assumption of~\eqref{eq: GDHAW 2}, there exist $x'_1,\ldots,x'_k \in \Delta^\mMp$ such that 
		\[ Z(x_i,x'_i) \fand \beta^\mM(x_i,x_{i+1}) \leq \beta^\mMp(x'_i,x'_{i+1}) \fand Z(x_{i+1},x'_{i+1}) \]
		for all $0 \leq i < k$. Since $\fand$ is associative \red{and due to~\eqref{fop: GDJSK 10}}, it follows that 
		\begin{eqnarray*}
			\!\!\!\!\!& \!\! & Z(x_0,x'_0) \fand \alpha^\mM(x_0,x_k) \\
			\!\!\!\!\!& \!=\! & Z(x_0,x'_0) \fand \beta^\mM(x_0,x_1) \fand\cdots\fand \beta^\mM(x_{k-1},x_k) \\ 
			\!\!\!\!\!& \!\leq\! & \beta^\mMp(x'_0,x'_1) \fand Z(x_1,x'_1) \fand \beta^\mM(x_1,x_2) \fand \cdots\fand \beta^\mM(x_{k-1},x_k) \\
			\!\!\!\!\!& \!\leq\! & \beta^\mMp(x'_0,x'_1) \fand \beta^\mMp(x'_1,x'_2) \fand Z(x_2,x'_2) \fand \beta^\mM(x_2,x_3) \fand \cdots\fand \beta^\mM(x_{k-1},x_k) \\
			\!\!\!\!\!& \!\leq\! & \ldots \\
			\!\!\!\!\!& \!\leq\! & \beta^\mMp(x'_0,x'_1) \fand\cdots\fand \beta^\mMp(x'_{k-1},x'_k) \fand Z(x_k,x'_k) \\ 
			\!\!\!\!\!& \!\leq\! & \alpha^\mMp(x'_0,x'_k) \fand Z(x_k,x'_k).
		\end{eqnarray*}
		Taking $y' = x'_k$, we obtain~\eqref{eq: HDHAK}.
	\end{itemize}
	
	The assertion~\eqref{eq: GDHAW 3} can be proved analogously as for~\eqref{eq: GDHAW 2}.
	
	Consider the assertion~\eqref{eq: GDHAW 1}. 
	The case when $\varphi = a$ is trivial. 
	The case when $\varphi = p$ follows from the condition~\eqref{eq: FB1}. 
	The case when $\varphi = \lnot \psi$ is reduced to the case when $\varphi = (\psi \to 0)$. 
	\begin{itemize}
		\item Case $\varphi = \psi \land \xi$: We have $\varphi^\mM(x) = \psi^\mM(x) \land \xi^\mM(x)$ and \mbox{$\varphi^\mMp(x') = \psi^\mMp(x') \land \xi^\mMp(x')$}. By the induction assumption of~\eqref{eq: GDHAW 1}, 
		\begin{eqnarray}
		Z(x,x') & \leq & \psi^\mM(x) \fequiv \psi^\mMp(x') \label{eq: JDGSH 1} \\
		Z(x,x') & \leq & \xi^\mM(x) \fequiv \xi^\mMp(x'). \label{eq: JDGSH 2}
		\end{eqnarray}
		By~\eqref{fop: GDJSK 120}, 
		\begin{eqnarray}
		\!\!\!\!\!\!\!\!\!\!&\!\!\!\!\!& 
		(\psi^\mM(x) \fequiv \psi^\mMp(x')) \land (\xi^\mM(x) \fequiv \xi^\mMp(x')) \leq (\varphi^\mM(x) \fequiv \varphi^\mMp(x')). \label{eq: JDGSH 3}
		\end{eqnarray}
		The assertion~\eqref{eq: GDHAW 1} follows from~\eqref{eq: JDGSH 1}, \eqref{eq: JDGSH 2} and~\eqref{eq: JDGSH 3}. 
		
		\item The case $\varphi = (\psi \lor \xi)$ is similar to the previous case, using~\eqref{fop: GDJSK 130} instead of~\eqref{fop: GDJSK 120}.  
		
		\item Case $\varphi = (\psi \to \xi)$ (and $\to\ \notin \Phi$): The proof is similar to the proof of the case when $\varphi = \psi \land \xi$, using~\eqref{fop: GDJSK 160} instead of~\eqref{fop: GDJSK 120}.  
		
		\item Case $\varphi = (a \to \psi)$: We have $\varphi^\mM(x) = (a \fto \psi^\mM(x))$ and \mbox{$\varphi^\mMp(x') = (a \fto \psi^\mMp(x'))$}. By the induction assumption of~\eqref{eq: GDHAW 1}, 
		\( Z(x,x') \leq (\psi^\mM(x) \fequiv \psi^\mMp(x')). \)
		The assertion~\eqref{eq: GDHAW 1} follows from this and~\eqref{fop: GDJSK 140}.  

		\item The case $\varphi = (\psi \to a)$ is similar to the previous case, using~\eqref{fop: GDJSK 150} instead of~\eqref{fop: GDJSK 140}.

		\item Case $\varphi = \tuple{\alpha}\psi$: Since $\mM$ is witnessed w.r.t.\ \fPDLwP, there exists $y \in \Delta^\mM$ such that 
		\begin{equation}
		\varphi^\mM(x) = \alpha^\mM(x,y) \fand \psi^\mM(y). \label{eq: HJKAQ 1}
		\end{equation}
		By the induction assumption \red{about}~\eqref{eq: GDHAW 2}, there exists $y' \in \Delta^\mMp$ such that 
		\begin{equation}
		Z(x,x') \fand \alpha^\mM(x,y) \leq \alpha^\mMp(x',y') \fand Z(y,y'). \label{eq: HJKAQ 2}
		\end{equation}
		By definition, 
		\begin{equation}
		\alpha^\mMp(x',y') \fand \psi^\mMp(y') \leq \varphi^\mMp(x'). \label{eq: HJKAQ 2a}
		\end{equation}
		
		By the induction assumption of~\eqref{eq: GDHAW 1}, 
		\begin{equation}
		Z(y,y') \ \leq\ \psi^\mM(y) \fequiv \psi^\mMp(y'). \label{eq: HJKAQ 3}
		\end{equation}
		
		By \eqref{eq: HJKAQ 2} and \eqref{fop: GDJSK 00},  
		\[ Z(x,x') \ \leq\ \alpha^\mM(x,y) \fto Z(y,y') \fand \alpha^\mMp(x',y'). \]
		By \eqref{eq: HJKAQ 3}, \eqref{fop: GDJSK 10} and \eqref{fop: GDJSK 20}, it follows that 
		\[ Z(x,x') \ \leq\ \alpha^\mM(x,y) \fto (\psi^\mM(y) \fequiv \psi^\mMp(y')) \fand \alpha^\mMp(x',y'). \]
		Since $\fand$ is commutative, by \eqref{fop: GDJSK 80} and \eqref{fop: GDJSK 20}, it follows that 
		\[ Z(x,x') \ \leq\ \alpha^\mM(x,y) \fto (\psi^\mM(y) \fequiv \alpha^\mMp(x',y') \fand \psi^\mMp(y')). \]
		By \eqref{fop: GDJSK 110}, it follows that 
		\[ Z(x,x') \ \leq\  \alpha^\mM(x,y) \fand \psi^\mM(y) \fto \alpha^\mMp(x',y') \fand \psi^\mMp(y'). \]
		By \eqref{eq: HJKAQ 1}, \eqref{eq: HJKAQ 2a} and \eqref{fop: GDJSK 20}, it follows that 
		\[ Z(x,x') \ \leq\ \varphi^\mM(x) \fto \varphi^\mMp(x'). \]
		Analogously, it can be shown that 
		\[ Z(x,x') \ \leq\ \varphi^\mMp(x') \fto \varphi^\mM(x). \]
		Therefore, 
		\[ Z(x,x') \ \leq\ \varphi^\mM(x) \fequiv \varphi^\mMp(x'). \]

		\item Case $\varphi = [\alpha]\psi$: Since $\mMp$ is witnessed w.r.t.\ \fPDLwP, there exists $y' \in \Delta^\mMp$ such that 
		\begin{equation}
		\varphi^\mMp(x') = (\alpha^\mMp(x',y') \fto \psi^\mMp(y')). \label{eq: KDNSJ 1}
		\end{equation}
		By the induction assumption \red{about}~\eqref{eq: GDHAW 3}, there exists $y \in \Delta^\mM$ such that 
		\begin{equation}
		Z(x,x') \fand \alpha^\mMp(x',y') \leq \alpha^\mM(x,y) \fand Z(y,y'). \label{eq: KDNSJ 2}
		\end{equation}
		By definition, 
		\begin{equation}
		\varphi^\mM(x) \ \leq\ \alpha^\mM(x,y) \fto \psi^\mM(y). \label{eq: KDNSJ 2a}
		\end{equation}
		By the induction assumption of~\eqref{eq: GDHAW 1}, 
		\begin{equation}
		Z(y,y') \ \leq\ \psi^\mM(y) \fequiv \psi^\mMp(y'). \label{eq: KDNSJ 3}
		\end{equation}
		
		By \eqref{eq: KDNSJ 2} and \eqref{fop: GDJSK 00}, 
		\[ Z(x,x') \ \leq\ \alpha^\mMp(x',y') \fto Z(y,y') \fand \alpha^\mM(x,y). \]
		By \eqref{eq: KDNSJ 3}, \eqref{fop: GDJSK 10} and \eqref{fop: GDJSK 20}, it follows that 
		\[ Z(x,x') \ \leq\ \alpha^\mMp(x',y') \fto (\psi^\mM(y) \fequiv \psi^\mMp(y')) \fand \alpha^\mM(x,y). \]
		Since $\fand$ is commutative, by \eqref{fop: GDJSK 70a} and \eqref{fop: GDJSK 20}, it follows that 
		\[ Z(x,x') \ \leq\ \alpha^\mMp(x',y') \fto ((\alpha^\mM(x,y) \fto \psi^\mM(y)) \fto \psi^\mMp(y')). \]
		By \eqref{fop: GDJSK 90}, it follows that 
		\[ Z(x,x') \ \leq\ (\alpha^\mM(x,y) \fto \psi^\mM(y)) \fto (\alpha^\mMp(x',y') \fto \psi^\mMp(y')). \]
		By \eqref{eq: KDNSJ 2a}, \eqref{eq: KDNSJ 1} and \eqref{fop: GDJSK 20}, it follows that 
		\[ Z(x,x') \ \leq\ \varphi^\mM(x) \fto \varphi^\mMp(x'). \]
		Analogously, it can be shown that 
		\[ Z(x,x') \ \leq\ \varphi^\mMp(x') \fto \varphi^\mM(x). \]
		Therefore, 
		\[ Z(x,x') \ \leq\ \varphi^\mM(x) \fequiv \varphi^\mMp(x'). \]
		This completes the proof.
		\myend
	\end{itemize}
\end{proof}

The following lemma is a counterpart of Lemma~\ref{lemma: GDHAW} for \fKz.

\begin{lemma} \label{lemma: GDHAW 2}
Let $\mM$ and $\mM'$ be fuzzy Kripke models that are witnessed w.r.t.\ \fKz and $Z$ a fuzzy bisimulation between $\mM$ and $ \mM'$. Then, the following \red{property holds} for every $x \in \Delta^\mM$, $x' \in \Delta^\mMp$ and every formula $\varphi$ of \fKz:
\[ Z(x,x') \ \leq\ \varphi^\mM(x) \fequiv \varphi^\mMp(x'). \]
\end{lemma}

This lemma can be proved analogously as done for the assertion~\eqref{eq: GDHAW 1} of Lemma~\ref{lemma: GDHAW}, by using~\eqref{eq: FB2} and \eqref{eq: FB3} instead of~\eqref{eq: GDHAW 2} and~\eqref{eq: GDHAW 3}, respectively. Roughly speaking, the proof is a simplification of the proof of Lemma~\ref{lemma: GDHAW}. 

\begin{remark}\label{remark: FDNBS}
Analyzing the proof of Lemma~\ref{lemma: GDHAW}, it can be seen that the condition~\eqref{req: HFJSK 2} ($\mL$ is a Heyting algebra if $\to\ \notin \Phi$ or $? \notin \Phi$) can be replaced by the conditions~\eqref{fop: GDJSK 160} and~\eqref{fop: GDJSK 170}. This also applies to Theorem~\ref{theorem: invariance}. 
\myend
\end{remark}

\begin{remark}
To justify that a condition like~\eqref{req: HFJSK 2} (or~\eqref{fop: GDJSK 160} and~\eqref{fop: GDJSK 170} together) is essential for Lemma~\ref{lemma: GDHAW} and Theorem~\ref{theorem: invariance}, we show that, if $\to\ \notin \Phi$ or $? \notin \Phi$, $L = [0,1]$ and $\fand$ is the {\L}ukasiewicz or product t-norm, then there exist finite fuzzy Kripke models $\mM$ and $\mM'$, a~fuzzy bisimulation $Z$ between $\mM$ and $\mM'$, $x \in \Delta^\mM$, $x' \in \Delta^\mMp$ and a formula $\varphi$ of \fPDLwP such that $Z(x,x') \not\leq (\varphi^\mM(x) \fequiv \varphi^\mMp(x'))$. Let
\begin{itemize}
\item $\SA = \emptyset$,\ \ $\SP = \{p,q\}$,\ \ $\varphi = (p \to q)$,\ \ $\psi = [p?]q$, 
\item $\Delta^\mM = \{v\}$,\ \ $p^\mM = \{v\!:\!0.2\}$,\ \ $q^\mM = \{v\!:\!0.2\}$,  
\item $\Delta^\mMp = \{v'\}$,\ \ $p^\mMp = \{v'\!:\!0.3\}$,\ \ $q^\mMp = \{v'\!:\!0.1\}$, 
\item $\fand$ be the {\L}ukasiewicz or product t-norm, 
\item $Z$ be the greatest fuzzy bisimulation between $\mM$ and $\mMp$.   
\end{itemize}
If $\fand$ is the {\L}ukasiewicz t-norm, then 
\begin{itemize}
\item $Z(v,v') = \min\{(p^\mM(v) \fequiv p^\mMp(v')), (q^\mM(v) \fequiv q^\mMp(v'))\} = 0.9$, 
\item $\varphi^\mM(v) = \psi^\mM(v) = 1$,\ \ $\varphi^\mMp(v') = \psi^\mMp(v') = 0.8$, 
\item $(\varphi^\mM(v) \fequiv \varphi^\mMp(v')) = (\psi^\mM(v) \fequiv \psi^\mMp(v')) = 0.8$. 
\end{itemize}
If $\fand$ is the product t-norm, then
\begin{itemize}
\item $Z(v,v') = \min\{(p^\mM(v) \fequiv p^\mMp(v')), (q^\mM(v) \fequiv q^\mMp(v'))\} = 0.5$, 
\item $\varphi^\mM(v) = \psi^\mM(v) = 1$,\ \ $\varphi^\mMp(v') = \psi^\mMp(v') = 1/3$. 
\item $(\varphi^\mM(v) \fequiv \varphi^\mMp(v')) = (\psi^\mM(v) \fequiv \psi^\mMp(v')) = 1/3$. 
\end{itemize}
Hence, $Z(v,v') \not\leq (\varphi^\mM(v) \fequiv \varphi^\mMp(v'))$ and $Z(v,v') \not\leq (\psi^\mM(v) \fequiv \psi^\mMp(v'))$.
\myend
\end{remark}

\begin{remark}\label{remark: HDJHS}\markRed
Consider the formula constructor $\varphi \,\&\, \psi$ whose meaning in a Kripke model $\mM$ is specified by $(\varphi \,\&\, \psi)^\mM(x) = \varphi^\mM(x) \fand \psi^\mM(x)$ for $x \in \Delta^\mM$. We show that formulas with this constructor may not be invariant under fuzzy bisimulations w.r.t.\ \fPDLwP for $\Phi = \{\cup, \to, ?\}$. 

Let $L = [0,1]$, $\SA = \emptyset$, $\SP = \{p\}$, \mbox{$\varphi = p \,\&\, p$} and let $\mM$ and $\mMp$ be Kripke models such that $\Delta^\mM = \{v\}$, $p^\mM(v) = 0.5$, $\Delta^\mMp = \{v'\}$ and $p^\mMp(v') = 1$. 
Consider the case when $\fand$ is the {\L}ukasiewicz or product t-norm. Observe that $Z: \Delta^\mM \times \Delta^\mMp \to L$ with $Z(v,v') = (0.5 \fequiv 1) = 0.5$ is a fuzzy bisimulation between $\mM$ and $\mMp$. We have $\varphi^\mMp(v') = 1$. 
If $\fand$ is the {\L}ukasiewicz t-norm, then $\varphi^\mM(v) = 0$ and $(\varphi^\mM(v) \fequiv \varphi^\mMp(v')) = 0$. 
If $\fand$ is the product t-norm, then $\varphi^\mM(v) = 0.25$ and $(\varphi^\mM(v) \fequiv \varphi^\mMp(v')) = 0.25$. 
Thus, $Z(v,v') \not\leq (\varphi^\mM(v) \fequiv \varphi^\mMp(v'))$. 
\myend
\end{remark}

\section{The Hennessy-Milner Property}
\label{sec: HM-prop}

In this section, we present and prove the Hennessy-Milner property of fuzzy bisimulations. 
It is formulated for the class of modally saturated models, which is larger than the class of image-finite models. 
Our notion of modal saturatedness is a counterpart of the ones given in~\cite{Fine75,BRV2001,FSS2020}.  

A fuzzy Kripke model $\mM$ is said to be {\em modally saturated} (w.r.t.~\fKz and the underlying residuated lattice $\mL$) if, for every $a \in L \setminus \{0\}$, every $x \in \Delta^\mM$, every $\varrho \in \SA$ and every infinite set $\Gamma$ of formulas in~\fKz, if for every finite subset $\Lambda$ of $\Gamma$ there exists $y \in \Delta^\mM$ such that $\varrho^\mM(x,y) \fand \varphi^\mM(y) \geq a$ for all $\varphi \in \Lambda$, then there exists $y \in \Delta^\mM$ such that $\varrho^\mM(x,y) \fand \varphi^\mM(y) \geq a$ for all $\varphi \in \Gamma$. 

\begin{proposition}
All image-finite fuzzy Kripke models are modally saturated.
\end{proposition}

\begin{proof}
Let $\mM$ be an image-finite fuzzy Kripke model, let $a \in L \setminus \{0\}$, $x \in \Delta^\mM$, $\varrho \in \SA$ and let $\Gamma$ be an infinite set of formulas in~\fKz. Assume that, for every finite subset $\Lambda$ of $\Gamma$, there exists $y \in \Delta^\mM$ such that $\varrho^\mM(x,y) \fand \varphi^\mM(y) \geq a$ for all $\varphi \in \Lambda$. For a contradiction, suppose that, for every $y \in \Delta^\mM$, there exists $\varphi_y \in \Gamma$ such that $\varrho^\mM(x,y) \fand \varphi_y^\mM(y) \not\geq a$. 
Let $\varphi_0$ be an arbitrary formula of $\Gamma$ and let $\Lambda = \{\varphi_y \mid \varrho^\mM(x,y) > 0\} \cup \{\varphi_0\}$. Since $\mM$ is image-finite, $\Lambda$ is finite. 
For every $y \in \Delta^\mM$, if $\varrho^\mM(x,y) = 0$, then by~\eqref{fop: GDJSK 40}, $\varrho^\mM(x,y) \fand \varphi_0^\mM(y) = 0 \not\geq a$, else $\varphi_y \in \Lambda$ and $\varrho^\mM(x,y) \fand \varphi_y^\mM(y) \not\geq a$. 
Hence, for every $y \in \Delta^\mM$, there exists $\varphi \in \Lambda$ such that $\varrho^\mM(x,y) \fand \varphi^\mM(y) \not\geq a$. This contradicts the assumption. 
\myend 
\end{proof}

Let $\mL$ be a complete residuated lattice. We say that the operator $\fand$ is {\em continuous} (w.r.t.\ infima) if, for every $x \in L$ and $Y \subseteq L$, $x \fand \inf Y = \inf \{x \fand y \mid y \in Y\}$. Clearly, all the G\"odel, {\L}ukasiewicz and product t-norms (\red{in particular} when $L$ is the unit interval $[0,1]$) are continuous. In addition, if $\mL$ is a Heyting algebra, then $\fand$ is continuous. 

\begin{theorem} \label{theorem: HM}
Let $\mM$ and $\mMp$ be fuzzy Kripke models that are witnessed w.r.t.\ \fKz and modally saturated.\footnote{These conditions are satisfied, for example, when $\mM$ and $\mMp$ are image-finite.} Suppose that the underlying residuated lattice \mbox{$\mL = \tuple{L, \leq, \fand, \fto, 0, 1}$} is complete and $\fand$ is continuous. Then, the fuzzy relation $Z : \Delta^\mM \times \Delta^\mMp \to L$ specified by 
\[ Z(x,x') = \inf\{\varphi^\mM(x) \fequiv \varphi^\mMp(x') \mid \textrm{$\varphi$ is a formula of \fKz}\} \] 
is the greatest fuzzy bisimulation between $\mM$ and~$\mMp$.	
\end{theorem}

\begin{proof}
By Lemma~\ref{lemma: GDHAW 2}, it is sufficient to prove that $Z$ is a fuzzy bisimulation between $\mM$ and $\mMp$. 

By definition, $Z$ satisfies the condition~\eqref{eq: FB1}.

We prove that $Z$ satisfies the condition \eqref{eq: FB2}. Let $\varrho \in \SA$, $x,y \in \Delta^\mM$ and $x' \in \Delta^\mMp$. Let $a = Z(x,x') \fand \varrho^\mM(x,y)$. For a contradiction, suppose that, for every $y' \in \Delta^\mMp$, \mbox{$a \not\leq \varrho^\mMp(x',y') \fand Z(y,y')$}. Since $\fand$ is continuous, by the definition of $Z(y,y')$, it follows that, for every $y' \in \Delta^\mMp$, there exists a formula $\varphi_{y'}$ of \fKz such that
\[ a \not\leq \varrho^\mMp(x',y') \fand (\varphi_{y'}^\mM(y) \fequiv \varphi_{y'}^\mMp(y')). \]
For every $y' \in \Delta^\mMp$, let 
\[ \psi_{y'} = (\varphi_{y'} \to \varphi_{y'}^\mM(y)) \land (\varphi_{y'}^\mM(y) \to \varphi_{y'}). \] 
Let $\Gamma = \{\psi_{y'} \mid y' \in \Delta^\mMp\}$.
Observe that, for every $y' \in \Delta^\mMp$, $\psi_{y'}^\mM(y) = 1$ (by~\eqref{fop: GDJSK 30}) and $a \not\leq \varrho^\mMp(x',y') \fand \psi_{y'}^\mMp(y')$. Since $\mMp$ is modally saturated, it follows that there exists a finite subset $\Psi$ of $\Gamma$ such that, for every $y' \in \Delta^\mMp$, there exists $\psi \in \Psi$ such that 
\begin{equation}
a \not\leq \varrho^\mMp(x',y') \fand \psi^\mMp(y'). \label{eq: HSKZK}
\end{equation}
Let $\varphi = \tuple{\varrho}\bigwedge\!\Psi$. It is a formula of \fKz. 
Thus, $\varphi^\mM(x) \geq \varrho^\mM(x,y)$ since $(\bigwedge\!\Psi)^\mM(y) = 1$. Since $\mMp$ is witnessed w.r.t.\ \fKz, by~\eqref{eq: HSKZK} and~\eqref{fop: GDJSK 10}, we have that $a \not\leq \varphi^\mMp(x')$, which means 
\[ Z(x,x') \fand \varrho^\mM(x,y) \not\leq \varphi^\mMp(x'). \]
Since $\varphi^\mM(x) \geq \varrho(x,y)$, by~\eqref{fop: GDJSK 10}, it follows that  
\[ Z(x,x') \fand \varphi^\mM(x) \not\leq \varphi^\mMp(x'). \]
By~\eqref{fop: GDJSK 00}, this implies that 
\[ Z(x,x') \not\leq (\varphi^\mM(x) \fto \varphi^\mMp(x')), \]
which contradicts the definition of $Z(x,x')$. 

Analogously, it can be proved that $Z$ satisfies the condition~\eqref{eq: FB3}. 
This completes the proof.
\myend
\end{proof}

\begin{corollary} \label{cor: HM 0}
Let $\mM$ be \red{an image-finite} fuzzy Kripke model. Suppose that the underlying residuated lattice \mbox{$\mL = \tuple{L, \leq, \fand, \fto, 0, 1}$} is \red{linear and} complete and $\fand$ is continuous. Then, the greatest fuzzy auto-bisimulation of $\mM$ exists and is a fuzzy equivalence relation.
\end{corollary}

\begin{proof}\markRed
Let \mbox{$Z : \Delta^\mM \times \Delta^\mM \to L$} be specified by 
\[ Z(x,x') = \inf\{\varphi^\mM(x) \fequiv \varphi^\mM(x') \mid \textrm{$\varphi$ is a formula of \fKz}\}. \] 
Since $\mM$ is image-finite, it is witnessed w.r.t.\ \fKz and modally saturated. By Theorem~\ref{theorem: HM}, $Z$ is the greatest fuzzy auto-bisimulation of~$\mM$. By Corollary~\ref{cor: HFHSJ}, $Z$ is a fuzzy equivalence relation. 
\myend
\end{proof}

\begin{corollary} \label{cor: HM}
Let $\mM$ and $\mMp$ be fuzzy Kripke models that are witnessed w.r.t.\ \fPDLwP and modally saturated. 
Suppose that the underlying residuated lattice \mbox{$\mL = \tuple{L, \leq, \fand, \fto, 0, 1}$} is complete, satisfies the conditions~\eqref{req: HFJSK 1} and~\eqref{req: HFJSK 2}, and $\fand$ is continuous. 
Then, for every $x \in \Delta^\mM$ and $x' \in \Delta^\mMp$, 
\begin{eqnarray*}
\!\!\!\!\!&&  
\inf\{\varphi^\mM(x) \fequiv \varphi^\mMp(x') \mid \textrm{$\varphi$ is a formula of \fKz}\} \\ 
\!\!\!\!\!&\!=\!& 
\inf\{\varphi^\mM(x) \fequiv \varphi^\mMp(x') \mid \textrm{$\varphi$ is a formula of \fPDLwP}\}. 
\end{eqnarray*}
\end{corollary}


\begin{proof}\markRed
Since \fKz is a sublanguage of \fPDLwP, it is sufficient to show that 
\begin{eqnarray*}
	\!\!\!\!\!&&  
	\inf\{\varphi^\mM(x) \fequiv \varphi^\mMp(x') \mid \textrm{$\varphi$ is a formula of \fKz}\} \\ 
	\!\!\!\!\!&\! \leq \!& 
	\inf\{\varphi^\mM(x) \fequiv \varphi^\mMp(x') \mid \textrm{$\varphi$ is a formula of \fPDLwP}\}. 
\end{eqnarray*}

Let \mbox{$Z : \Delta^\mM \times \Delta^\mMp \to L$} be specified by 
\[ Z(x,x') = \inf\{\varphi^\mM(x) \fequiv \varphi^\mMp(x') \mid \textrm{$\varphi$ is a formula of \fKz}\}. \] 
By Theorem~\ref{theorem: HM}, $Z$ is the greatest fuzzy bisimulation between $\mM$ and~$\mMp$.	
Let $x \in \Delta^\mM$, $x' \in \Delta^\mMp$ and let $\varphi$ be an arbitrary formula of \fPDLwP. 
By Lemma~\ref{lemma: GDHAW}, \mbox{$Z(x,x') \leq (\varphi^\mM(x) \fequiv \varphi^\mMp(x'))$}. 
Therefore, 
\[
	Z(x,x') \leq  \inf\{\varphi^\mM(x) \fequiv \varphi^\mMp(x') \mid \textrm{$\varphi$ is a formula of \fPDLwP}\}.
\]
This completes the proof. 
\myend
\end{proof}

\section{Related Work \red{and Discussion}}
\label{sec: rw}

The works~\cite{ai/FanL14,Fan15,FSS2020} on fuzzy bisimulations have been briefly discussed in the introduction. 
We give below some additional remarks on these works before discussing other works related to logical characterizations of fuzzy/crisp bisimulations or simulations. 

The results on fuzzy bisimulations of~\cite{ai/FanL14} are formulated only for finite social networks over the residuated lattice $[0,1]$ using the G\"odel t-norm. In that work, Fan and Liau did consider extending their results on fuzzy bisimulations to the settings with the {\L}ukasiewicz and product t-norms. However, in~\cite[Example~2]{ai/FanL14} they claimed that the extension does not work. The problem with that claim is that the authors used the logical language with the additional conjunction~$\&$ which is interpreted as~$\fand$ \red{(see Remark~\ref{remark: HDJHS})}. 
In \cite{ai/FanL14} Fan and Liau also studied crisp bisimulations under the name ``generalized regular equivalence relations'' for finite weighted social networks. They provided logical characterizations for crisp bisimulations under the G\"odel, {\L}ukasiewicz and product semantics. The characterizations are formulated w.r.t.\ fuzzy multimodal logics possibly with converse, which are extended with involutive negation and/or the Baaz projection operator. They concern invariance of modal formulas under crisp bisimulations and the Hennessy-Milner property of crisp bisimulations.

In~\cite{Fan15} Fan also studied crisp bisimulations for fuzzy monomodal logics under the G\"odel semantics. She provided logical characterizations of such bisimulations in the basic fuzzy monomodal logics possibly with converse, which are extended with involutive negation and/or the Baaz projection operator. The results of~\cite{Fan15} on invariance of modal formulas and the Hennessy-Milner property for both crisp and fuzzy bisimulations are formulated for image-finite fuzzy Kripke models over a~signature with only one accessibility relation. 

In~\cite{FSS2020} Nguyen {\em et al}.\ also provided logical characterizations of crisp bisimulations for fuzzy description logics under the G{\"o}del semantics. For the case with such bisimulations, the considered logics are extended with the Baaz projection operator or involutive negation. Apart from results on invariance of concepts and the Hennessy-Milner property of crisp/fuzzy bisimulations, the work~\cite{FSS2020} also gives results on conditional invariance of TBoxes and ABoxes under crisp/fuzzy bisimulations, separation of the expressive power of fuzzy description logics, and minimization of fuzzy interpretations by using crisp bisimulations.

\red{When restricting to invariance results and the Hennessy-Milner property of fuzzy bisimulations in logics that are sublogics of \fPDL, our results are more general than the results of \cite{ai/FanL14,Fan15,FSS2020}. For example, Theorem~\ref{theorem: HM} together with Corollary~\ref{cor: HM} is strictly more general than Theorem~3 of~\cite{Fan15}. One of the reasons is that Theorem~\ref{theorem: HM} and Corollary~\ref{cor: HM} are formulated for complete residuated lattices, while Theorem~3 of~\cite{Fan15} is formulated only for the lattice $[0,1]$ using the G\"odel t-norm. The other reasons are that: the logic considered in Theorem~3 of~\cite{Fan15} is a proper sublogic of \fPDLwP; Theorem~\ref{theorem: HM} and Corollary~\ref{cor: HM} are formulated for witnessed and modally saturated models, while Theorem~3 of~\cite{Fan15} is formulated only for image-finite models. On the other hand, all the papers \cite{ai/FanL14,Fan15,FSS2020} contain results formulated for modal/description logics with converse/inverse and/or the Baaz projection operator, which are not studied in the current work. The paper \cite{FSS2020} also allows other constructors of description logics.}

In~\cite{DBLP:journals/fss/WuD16,DBLP:journals/ijar/WuCHC18,DBLP:journals/fss/WuCBD18}, Wu {\em et al}.\ provided logical characterizations of crisp bisimulations/simulations for a few variants of fuzzy transition systems. The results are formulated w.r.t.\ crisp Hennessy-Milner logics, which use values from the unit interval $[0,1]$ as thresholds for modal operators.

In \cite{DBLP:journals/ijar/PanC0C14} Pan {\em et at.}\ provided logical characterizations of fuzzy simulations for finite fuzzy labeled transition systems over finite residuated lattices. They are formulated w.r.t.\ an existential Hennessy-Milner logic. 
In~\cite{DBLP:journals/ijar/PanLC15} Pan {\em et at.}\ provided logical characterizations of simulations for finite quantitative transition systems over finite Heyting algebras. Quantitative transition systems are transition systems without labels for states but extended with a fuzzy equality relation between actions. Simulations studied in~\cite{DBLP:journals/ijar/PanLC15} are either fuzzy simulations or crisp simulations parameterized by a threshold used as a cut for the fuzzy equality relation between actions. The logical characterizations of simulations provided in~\cite{DBLP:journals/ijar/PanLC15} are formulated w.r.t.\ an existential cut-based crisp Hennessy-Milner logic for the case of crisp simulations, and w.r.t.\ an existential fuzzy Hennessy-Milner logic for the case of fuzzy simulations. 

In~\cite{Nguyen-TFS2019} we provided logical characterizations of crisp cut-based simulations and bisimilarity for a large class of fuzzy description logics under the Zadeh semantics. The results concern preservation of information by such simulations, conditional invariance of ABoxes and TBoxes under bisimilarity between witnessed interpretations, as well as the Hennessy-Milner property for fuzzy description logics under the Zadeh semantics.

\red{In~\cite{aml/MartiM18} Marti and Metcalfe studied logical characterizations of {\em crisp} bisimulations in chain-based modal logics. The considered logics are monomodal logics whose formulas are interpreted in many-valued Kripke models over a chain-based algebra with a crisp frame. A chain-based algebra is a linear and complete bounded lattice. The main results of~\cite{aml/MartiM18} concern characterizations of classes of Kripke models that have (resp.\ do not have) the Hennessy-Milner property.}  

\red{The work~\cite{fuin/Diaconescu20} by Diaconescu concerns logical characterizations of {\em crisp} bisimulations in fuzzy modal logics over complete MTL-chains. A complete MTL-chain is a linear and complete residuated lattice. The considered logics are fuzzy monomodal logics that allow many-valued formulas and accessibility relations. The main result of~\cite{fuin/Diaconescu20} gives a necessary and sufficient algebraic condition for the class of image-finite Kripke models for such logics to admit the Hennessy-Milner property.}

\section{Conclusions}
\label{sec: conc}

We have provided and proved logical characterizations of fuzzy bisimulations in the fuzzy propositional dynamic logic \fPDL and its sublogics over residuated lattices. The results concern invariance of formulas under fuzzy bisimulations and the Hennessy-Milner property of fuzzy bisimulations. The first theorem is formulated for fuzzy Kripke models that are witnessed, whereas the second theorem is formulated for fuzzy Kripke models that are witnessed and modally saturated. 

Our results can be reformulated for other fuzzy structures such as fuzzy labeled transition systems and fuzzy interpretations in description logics. 
It is worth emphasizing that our results concern fuzzy bisimulations over {\em general residuated lattices}. They are interesting from the theoretical point of view, as the previous results on fuzzy bisimulations are formulated and proved only for the residuated lattice $[0,1]$ using the G\"odel t-norm or Heyting algebras. 

In certain applications, the product t-norm is more suitable than the G\"odel t-norm. For example, the closeness of a person to his/her great-grandmother can be assumed to be smaller than the closeness of that person to his/her mother. Furthermore, the product residuum is continuous w.r.t.\ both the arguments, whereas the G\"odel residuum is not. This causes that the product residuum is more resistant to noise than the G\"odel residuum. 
Our logical characterizations of fuzzy bisimulations open the way for studying logical similarity between individuals and concept learning in fuzzy description logics under the product semantics by applying fuzzy bisimulations.  

On the technical matters, our results are formulated on a general level. 
Residuated lattices considered in this work may be infinite, whereas the work~\cite{DBLP:journals/ijar/PanC0C14} considers only finite residuated lattices. The class of fuzzy Kripke models that are witnessed and modally saturated is larger than the class of image-finite fuzzy Kripke models studied in~\cite{EleftheriouKN12,Fan15}, the class of finite weighted social networks studied in~\cite{ai/FanL14} and the class of finite fuzzy labeled transition systems studied in~\cite{DBLP:journals/ijar/PanC0C14,DBLP:journals/ijar/PanLC15}. The considered fuzzy logic \fPDL contains the program constructors of propositional dynamic logic, which are absent in~\cite{EleftheriouKN12,ai/FanL14,Fan15,DBLP:journals/ijar/PanC0C14,DBLP:journals/ijar/PanLC15}. They correspond to role constructors in description logics.

{\markRed
\section*{Acknowledgments}

The author would like to thank the anonymous reviewers for very helpful comments. 
}


\bibliography{BSfDL}
\bibliographystyle{elsarticle-harv}

\appendix

\section{Proof of Lemma~\ref{lemma: JHFJW}}
\label{appendix A}

	Note that $\fand$ is commutative and associative. Let $x,x',y,y',z \in L$. 
	\begin{itemize}
		\item Since $\fand$ is commutative, to prove~\eqref{fop: GDJSK 10}, it is sufficient to show that, if $x \leq x'$, then $x \fand y \leq x' \fand y$. Assume that $x \leq x'$. By~\eqref{fop: GDJSK 00}, $x' \leq (y \fto (x' \fand y))$. Hence, $x \leq (y \fto (x' \fand y))$. By~\eqref{fop: GDJSK 00}, it follows that $x \fand y \leq x' \fand y$, which completes the proof of~\eqref{fop: GDJSK 10}.  
		
		\item Consider the assertion~\eqref{fop: GDJSK 20} and assume that $x' \leq x$ and $y \leq y'$. By~\eqref{fop: GDJSK 10}, $(x \fto y) \fand x' \leq$ \mbox{$(x \fto y) \fand x$}. By~\eqref{fop: GDJSK 00}, $(x \fto y) \fand x \leq y$. Hence, $(x \fto y) \fand x' \leq y$, which implies $(x \fto y) \leq$ $(x' \fto y)$ by using~\eqref{fop: GDJSK 00}. 
		By~\eqref{fop: GDJSK 00}, $(x' \fto y) \fand x'\leq y$. Since $y \leq y'$, it follows that $(x' \fto y) \fand x'\leq y'$, which implies $(x' \fto y) \leq (x' \fto y')$ by using~\eqref{fop: GDJSK 00}. We have proved that $(x \fto y) \leq (x' \fto y)$ and $(x' \fto y) \leq (x' \fto y')$, which together imply~\eqref{fop: GDJSK 20}. 
		
		\item Consider the assertion~\eqref{fop: GDJSK 30}. We have that $x \leq y$ iff $1 \fand x \leq y$ iff $1 \leq (x \fto y)$, by using~\eqref{fop: GDJSK 00}. The last inequality implies~\eqref{fop: GDJSK 30}.  
		
		\item Since $0 \leq (x \fto 0)$, by~\eqref{fop: GDJSK 00}, $0 \fand x \leq 0$. Hence, $x \fand 0 = 0 \fand x \leq 0$ and the assertion~\eqref{fop: GDJSK 40} holds. 
		
		\item Consider the assertion~\eqref{fop: GDJSK 50}. By~\eqref{fop: GDJSK 10}, $x \fand y \leq x \fand (y \lor z)$ and $x \fand z \leq x \fand (y \lor z)$. Hence, $x \fand y \lor x \fand z \leq x \fand (y \lor z)$. It remains to prove the converse. By~\eqref{fop: GDJSK 00}, $y \leq (x \fto x \fand y)$. By~\eqref{fop: GDJSK 20}, it follows that $y \leq (x \fto x \fand y \lor x \fand z)$. Similarly, it can be shown that $z \leq (x \fto x \fand y \lor x \fand z)$. Hence, $y \lor z \leq (x \fto x \fand y \lor x \fand z)$. By~\eqref{fop: GDJSK 00}, it follows that $x \fand (y \lor z) \leq x \fand y \lor x \fand z$. This completes the proof of~\eqref{fop: GDJSK 50}. 
		
		\item The assertion~\eqref{fop: GDJSK 60} follows from~\eqref{fop: GDJSK 00} and the commutativity of~$\fand$.
		
		\item Consider the assertions~\eqref{fop: GDJSK 70} and~\eqref{fop: GDJSK 70a}. By~\eqref{fop: GDJSK 60} and~\eqref{fop: GDJSK 10}, $x \fand (x \fto y) \fand (y \fto z) \leq y \fand (y \fto z) \leq z$. 
		Hence, $x \fand (y \fto z) \fand (x \fto y) \leq z$. This implies~\eqref{fop: GDJSK 70}, by using~\eqref{fop: GDJSK 00}. The assertion~\eqref{fop: GDJSK 70a} follows from~\eqref{fop: GDJSK 70}, by using~\eqref{fop: GDJSK 10}. 
		
		\item Consider the assertion~\eqref{fop: GDJSK 80}. We need to prove that 
		\begin{eqnarray}
		\!\!\!\!\!\!\!\!\!\! x \fand (y \fequiv z) & \leq & (y \fto x \fand z) \label{fop: GDJSK 80a} \\
		\!\!\!\!\!\!\!\!\!\! x \fand (y \fequiv z) & \leq & (x \fand z \fto y) \label{fop: GDJSK 80b}
		\end{eqnarray}
		By~\eqref{fop: GDJSK 00}, $(y \fto z) \fand y \leq z$. By~\eqref{fop: GDJSK 10}, it follows that $x \fand (y \fto z) \fand y \leq x \fand z$. By~\eqref{fop: GDJSK 00}, it follows that $x \fand (y \fto z) \leq (y \fto x \fand z)$. This implies~\eqref{fop: GDJSK 80a}, by using~\eqref{fop: GDJSK 10}. Consider the assertion~\eqref{fop: GDJSK 80b}. By~\eqref{fop: GDJSK 00}, $(z \fto y) \fand z \leq y$. By~\eqref{fop: GDJSK 10}, it follows that $x \fand (z \fto y) \fand x \fand z \leq y$. By~\eqref{fop: GDJSK 00}, it follows that $x \fand (z \fto y) \leq (x \fand z \fto y)$. This implies~\eqref{fop: GDJSK 80b}, by using~\eqref{fop: GDJSK 10}.   
		
		\item Consider the assertions~\eqref{fop: GDJSK 90}--\eqref{fop: GDJSK 110}. By~\eqref{fop: GDJSK 00}, $(x \fto (y \fto z)) \fand x \leq (y \fto z)$ and $(y \fto z) \fand y \leq z$. By~\eqref{fop: GDJSK 10}, it follows that 
		\begin{equation}
		(x \fto (y \fto z)) \fand y \fand x \leq z. \label{eq: JROSJ}
		\end{equation}
		By applying~\eqref{fop: GDJSK 00} twice, this implies \mbox{$(x \fto (y \fto z)) \leq (y \fto (x \fto z))$}, which in turn implies~\eqref{fop: GDJSK 90}. The assertion~\eqref{eq: JROSJ} also implies~\eqref{fop: GDJSK 100}, by using~\eqref{fop: GDJSK 00} and the commutativity of~$\fand$. The assertion~\eqref{fop: GDJSK 110} follows from~\eqref{fop: GDJSK 100}, by using~\eqref{fop: GDJSK 20}.  
		
		\item Consider the assertion~\eqref{fop: GDJSK 115}. By~\eqref{fop: GDJSK 60} and~\eqref{fop: GDJSK 10}, we have that $x \fand (x \fto y) \fand (y \fto z) \leq z$. The assertion~\eqref{fop: GDJSK 115} follows from this, using~\eqref{fop: GDJSK 00} and the commutativity of~$\fand$.   
		
		\item The assertion~\eqref{fop: GDJSK 115a} follows from~\eqref{fop: GDJSK 115}, using~\eqref{fop: GDJSK 10} and the commutativity of~$\fand$ and~$\fequiv$.  
		
		\item Due to the commutativity of $\fequiv$, to prove~\eqref{fop: GDJSK 120} it is sufficient to show that 
		\[ (x \fequiv x') \land (y \fequiv y') \leq (x \land y \fto x' \land y'). \] 
		By~\eqref{fop: GDJSK 00}, this is equivalent to $((x \fequiv x') \land (y \fequiv y')) \fand (x \land y) \leq x' \land y'$. We need to prove that 
		\begin{eqnarray}
		\!\!\!\!\!\!\!\!\!\! ((x \fequiv x') \land (y \fequiv y')) \fand (x \land y) & \leq & x' \label{fop: GDJSK x1} \\
		\!\!\!\!\!\!\!\!\!\! ((x \fequiv x') \land (y \fequiv y')) \fand (x \land y) & \leq & y'. \label{fop: GDJSK x2}
		\end{eqnarray}
		By~\eqref{fop: GDJSK 00}, $(x \fto x') \fand x \leq x'$. This implies~\eqref{fop: GDJSK x1}, by using~\eqref{fop: GDJSK 10}. Similarly, \eqref{fop: GDJSK x2} also holds. 
		
		\item Due to the commutativity of $\fequiv$, to prove~\eqref{fop: GDJSK 130} it is sufficient to show that 
		\[ (x \fequiv x') \land (y \fequiv y') \leq (x \lor y \fto x' \lor y'). \] 
		By~\eqref{fop: GDJSK 00}, this is equivalent to $((x \fequiv x') \land (y \fequiv y')) \fand (x \lor y) \leq x' \lor y'$. By~\eqref{fop: GDJSK 50}, it is sufficient to prove that 
		\begin{eqnarray}
		\!\!\!\!\!\!\!\!\!\! ((x \fequiv x') \land (y \fequiv y')) \fand x & \leq & x' \label{fop: GDJSK x3} \\
		\!\!\!\!\!\!\!\!\!\! ((x \fequiv x') \land (y \fequiv y')) \fand y & \leq & y'. \label{fop: GDJSK x4}
		\end{eqnarray}
		By~\eqref{fop: GDJSK 00}, $(x \fto x') \fand x \leq x'$. This implies~\eqref{fop: GDJSK x3}, by using~\eqref{fop: GDJSK 10}. Similarly, \eqref{fop: GDJSK x4} also holds.
		
		\item To prove~\eqref{fop: GDJSK 140}, it is sufficient to show that $(x \fto y) \leq ((z \fto x) \fto (z \fto y))$. By~\eqref{fop: GDJSK 00}, this is equivalent to $(x \fto y) \fand (z \fto x) \fand z \leq y$. This latter inequality holds because, by~\eqref{fop: GDJSK 60} and~\eqref{fop: GDJSK 10}, 
		\[ z \fand (z \fto x) \fand (x \fto y) \leq x \fand (x \fto y) \leq y. \]
		
		\item To prove~\eqref{fop: GDJSK 150}, it is sufficient to show that $(y \fto x) \leq ((x \fto z) \fto (y \fto z))$. By~\eqref{fop: GDJSK 00}, this is equivalent to $(y \fto x) \fand (x \fto z) \fand y \leq z$. This latter inequality holds because, by~\eqref{fop: GDJSK 60} and~\eqref{fop: GDJSK 10}, 
		\[ y \fand (y \fto x) \fand (x \fto z) \leq x \fand (x \fto z) \leq z. \]
		
		
		\item Consider the assertion~\eqref{fop: GDJSK 160} and suppose that $\mL$ is a Heyting algebra. Due to the commutativity of $\fequiv$, to prove~\eqref{fop: GDJSK 160} it is sufficient to prove that 
		\[ (x \fequiv x') \fand (y \fequiv y') \leq ((x \fto y) \fto (x' \fto y')). \] 
		By~\eqref{fop: GDJSK 10} and~\eqref{fop: GDJSK 00}, it is sufficient to prove that 
		$(x' \fto x) \fand (y \fto y') \fand (x \fto y) \fand x' \leq y'$. 
		This holds because, by~\eqref{fop: GDJSK 60} and~\eqref{fop: GDJSK 10},  
		\[ x'  \fand (x' \fto x) \fand (x \fto y) \fand (y \fto y') \leq x \fand (x \fto y) \fand (y \fto y') \leq y \fand (y \fto y') \leq y'. \]
		
		\item Consider the assertion~\eqref{fop: GDJSK 170} and suppose that $\mL$ is a Heyting algebra and $x \leq (y \fequiv z)$. Since $x \leq (y \fequiv z)$, we have that $x \leq (y \fto z)$. By~\eqref{fop: GDJSK 00}, it follows that $x \land y = x \fand y \leq z$. Hence, $x \land y \leq x \land z$ since $\land$ is idempotent. Similarly, it can also be shown that $x \land z \leq x \land y$. Therefore, $x \land y = x \land z$, which means $x \fand y = x \fand z$.
	\end{itemize}


\section{The Relationship with Fuzzy Bisimulations between Fuzzy Automata}
\label{appendix B}

Clearly, a fuzzy Kripke model can be treated as a fuzzy labeled transition system (\FLTS) and Definition~\ref{def: HDJAK} (which specifies fuzzy bisimulations) can be applied to \FLTSs. 
In~\cite{CiricIDB12}, {\'C}iri{\'c} {\em et al}.\ introduced a few kinds of fuzzy bisimulations (and simulations) for fuzzy automata over complete residuated lattices. Among them the one that researchers would have in mind as the default is called ``forward bisimulation''. We recall it below and simply refer to it as fuzzy bisimulation between fuzzy automata. After that we relate it to the notion of fuzzy bisimulation between fuzzy Kripke models.

In this appendix, suppose that the underlying residuated lattice $\mL$ is complete. 

Given fuzzy sets $R: X \to L$, $S: Y \to L$ and $Z: X \times Y \to L$, we define $(R \circ Z) : Y \to L$ and $(Z \circ S) : X \to L$ to be the fuzzy sets such that
\begin{eqnarray*}
(R \circ Z)(y) & = & \sup \{R(x) \fand Z(x,y) \mid x \in X\} \quad \textrm{for } y \in Y; \\
(Z \circ S)(x) & = & \sup \{Z(x,y) \fand S(y) \mid y \in Y\} \quad \textrm{ for } x \in X.
\end{eqnarray*}

A {\em fuzzy automaton} over an alphabet $\Sigma$ (and $\mL$) is a tuple $\mA = \tuple{A, \delta^\mA, \sigma^\mA, \tau^\mA}$, where $A$ is a non-empty set of states, $\delta^\mA : A \times \Sigma \times A \to L$ is the fuzzy transition function, $\sigma^\mA : A \to L$ is the fuzzy set of initial states, and $\tau^\mA : A \to L$ is the fuzzy set of terminal states. 
For $\varrho \in \Sigma$, by $\delta_\varrho^\mA$ we denote the fuzzy relation on $A$ such that $\delta_\varrho^\mA(x,y) = \delta^\mA(x,\varrho,y)$ for $x,y \in A$. 

A fuzzy automaton $\mA\ \red{= \tuple{A, \delta^\mA, \sigma^\mA, \tau^\mA}}$ is {\em image-finite} if:
\begin{itemize}
\item \red{the set $\{x \in A \mid \sigma^\mA(x) > 0\}$ is finite, and} 
\item for every $\varrho \in \Sigma$ and every $x \in A$, the set $\{y \in A \mid \delta_\varrho^\mA(x,y) > 0\}$ is finite.
\end{itemize}

Given fuzzy automata $\mA = \tuple{A, \delta^\mA, \sigma^\mA, \tau^\mA}$ and \red{$\mAp = \tuple{A', \delta^\mAp, \sigma^\mAp, \tau^\mAp}$} over an alphabet~$\Sigma$, a~{\em fuzzy bisimulation} (called ``forward bisimulation'' in~\cite{CiricIDB12}) between $\mA$ and \red{$\mAp$} is a fuzzy relation \mbox{$Z: A \times \red{A'} \to L$} satisfying the following conditions for all $\varrho \in \Sigma$:
\begin{eqnarray}
\sigma^\mA & \leq & \sigma^{\red{\mAp}} \circ Z^- \label{eq: HFHAJ 1} \\
Z^- \circ \delta_\varrho^\mA & \leq & \delta_\varrho^{\red{\mAp}} \circ Z^- \label{eq: HFHAJ 2} \\
Z^- \circ \tau^\mA & \leq & \tau^{\red{\mAp}} \label{eq: HFHAJ 3} \\[1ex]
\sigma^{\red{\mAp}} & \leq & \sigma^\mA \circ Z \label{eq: HFHAJ 4} \\
Z \circ \delta_\varrho^{\red{\mAp}} & \leq & \delta_\varrho^\mA \circ Z \label{eq: HFHAJ 5} \\
Z \circ \tau^{\red{\mAp}} & \leq & \tau^\mA. \label{eq: HFHAJ 6} 
\end{eqnarray}

Given a fuzzy automaton $\mA = \tuple{A,\delta^\mA,\sigma^\mA,\tau^\mA}$ over an alphabet $\Sigma$, we define the {\em fuzzy Kripke model corresponding to \red{$\mA$}} to be the fuzzy Kripke model $\mM$ over the signature $\tuple{\SA,\SP}$ with $\SA = \Sigma$ and $\SP = \{i,f\}$ such that:
\begin{itemize}
	\item $\Delta^\mM = A \cup \{s_i,s_f\}$, where $s_i$ and $s_f$ are new states;
	\item $i^\mM = \{s_i\!:\!1\}$ and $f^\mM = \{s_f\!:\!1\}$;
	\item for every $\varrho \in \SA$, $x,y \in A$ and $z \in \Delta^\mM$: 
	\begin{itemize} 
		\item $\varrho^\mM(x,y) = \delta^\mA(x,\varrho,y)$, 
		\item $\varrho^\mM(s_i,x) = \sigma^\mA(x)$ and $\varrho^\mM(x,s_f) = \tau^\mA(x)$, 
		\item $\varrho^\mM(z,s_i) = \varrho^\mM(s_f,z) = \varrho^\mM(s_i,s_f) = 0$. 
	\end{itemize}
\end{itemize}
Thus, $s_i$ (resp.\ $s_f$) stands for the new unique initial (resp.\ terminal) state; the propositions $i$ and $f$ are used to identify $s_i$ and~$s_f$, respectively. 
The given definition is a counterpart of the definition of the fuzzy interpretation (in description logic) that corresponds to a~fuzzy automaton~\cite{TFS2020}. 

Recall that, in this appendix, the underlying residuated lattice is assumed to be complete. 
The following proposition relates our notion of fuzzy bisimulation between fuzzy Kripke models to the notion of fuzzy bisimulation between fuzzy automata, which is defined and called ``forward bisimulation'' by {\'C}iri{\'c} {\em et al.}~\cite{CiricIDB12}. 

\begin{proposition}\label{prop: HFKSS 2}
	Let $\mA = \tuple{A, \delta^\mA, \sigma^\mA, \tau^\mA}$ and $\mAp = \tuple{A', \delta^\mAp, \sigma^\mAp, \tau^\mAp}$ be fuzzy automata over the same alphabet, $\mM$ and $\mMp$ the fuzzy Kripke models corresponding to $\mA$ and $\mAp$, respectively. 
	Let $s_i,s_f \in \Delta^\mM$ and $s'_i, s'_f \in \Delta^\mMp$ be the states such that $i^\mM(s_i) = f^\mM(s_f) = i^\mMp(s'_i) = f^\mMp(s'_f) = 1$. Let $Z$ be a fuzzy relation between $A$ and $A'$, $Z_2$ the fuzzy relation between $\Delta^\mM$ and $\Delta^\mMp$ such that $Z_2 = Z \cup \{\tuple{s_i,s'_i}\!:\!1, \tuple{s_f,s'_f}\!:\!1\}$. 
	\begin{itemize}
		\item If $Z_2$ is a fuzzy bisimulation between $\mM$ and~$\mMp$, then $Z$ is a fuzzy bisimulation between $\mA$ and $\mAp$.
		\item Conversely, if the underlying residuated lattice is also linear, $\mA$ and $\mAp$ are image-finite and $Z$ is a fuzzy bisimulation between~$\mA$ and~$\mAp$, then $Z_2$ is a fuzzy bisimulation between $\mM$ and~$\mMp$.
	\end{itemize}
\end{proposition}

This proposition uses conditions similar to the ones of Proposition~\ref{prop: HFKSS}. The reason is that the definition of fuzzy bisimulations between fuzzy automata uses conditions similar to \eqref{eq: ERJSJ 1}--\eqref{eq: ERJSJ 3}. 

\begin{proof}\markRed
Suppose that $Z_2$ is a fuzzy bisimulation between $\mM$ and~$\mMp$. We show that $Z$ is a fuzzy bisimulation between $\mA$ and $\mAp$. Let $\varrho \in \Sigma$. We prove that $Z$ satisfies the conditions~\eqref{eq: HFHAJ 1}--\eqref{eq: HFHAJ 3}. The proof of that $Z$ satisfies the conditions~\eqref{eq: HFHAJ 4}--\eqref{eq: HFHAJ 6} is similar and omitted.  

Consider the condition~\eqref{eq: HFHAJ 1} and let $y \in A$. We need to prove that 
\begin{equation}
\sigma^\mA(y) \leq (\sigma^\mAp \circ Z^-)(y). \label{eq: NCBAA} 
\end{equation}
Since $Z_2$ is a fuzzy bisimulation between $\mM$ and~$\mMp$, by the condition~\eqref{eq: FB2} for $Z_2$, 
\[ \E y' \in \Delta^\mMp\ (Z_2(s_i,s'_i) \fand \varrho^\mM(s_i,y) \leq \varrho^\mMp(s'_i,y') \fand Z_2(y,y')). \]
By the assumptions about $Z_2$, $\mM$ and $\mM'$, this means that there exists $y' \in \Delta^\mMp$ such that  
\begin{equation}
\sigma^\mA(y) \leq \varrho^\mMp(s'_i,y') \fand Z_2(y,y'). \label{eq: HKAJA}
\end{equation}
Since $y \in A$, if $y' \notin A'$, then $Z_2(y,y') = 0$, and by~\eqref{fop: GDJSK 40}, \eqref{eq: HKAJA} implies~\eqref{eq: NCBAA}. If $y' \in A'$, then $\varrho^\mMp(s'_i,y') = \sigma^\mAp(y')$, $Z_2(y,y') = Z(y,y')$ and \eqref{eq: HKAJA} also implies~\eqref{eq: NCBAA}.

Consider the condition~\eqref{eq: HFHAJ 2}. Let $x' \in A'$ and $y \in A$. We need to prove that 
\[ (Z^- \circ \delta_\varrho^\mA)(x',y) \leq (\delta_\varrho^\mAp \circ Z^-)(x',y). \]
Let $x$ be an arbitrary element of $A$. It is sufficient to show that  
\begin{equation}
Z(x,x') \fand \delta_\varrho^\mA(x,y) \leq (\delta_\varrho^\mAp \circ Z^-)(x',y). \label{eq: HDJAA}
\end{equation}
Since $Z_2$ is a fuzzy bisimulation between $\mM$ and~$\mMp$, by the condition~\eqref{eq: FB2} for $Z_2$, 
\[ \E y' \in \Delta^\mMp\ (Z_2(x,x') \fand \varrho^\mM(x,y) \leq \varrho^\mMp(x',y') \fand Z_2(y,y')). \]
By the assumptions about $Z_2$, $\mM$ and $\mM'$, this means that there exists $y' \in \Delta^\mMp$ such that  
\begin{equation}
Z(x,x') \fand \delta_\varrho^\mA(x,y) \leq \varrho^\mMp(x',y') \fand Z_2(y,y'). \label{eq: HGJAQ}
\end{equation}
Since $y \in A$, if $y' \notin A'$, then $Z_2(y,y') = 0$, and by~\eqref{fop: GDJSK 40}, \eqref{eq: HGJAQ} implies~\eqref{eq: HDJAA}. If $y' \in A'$, then $\varrho^\mMp(x',y') = \delta_\varrho^\mAp(x',y')$, $Z_2(y,y') = Z(y,y')$ and \eqref{eq: HGJAQ} also implies~\eqref{eq: HDJAA}.

Consider the condition~\eqref{eq: HFHAJ 3} and let $x' \in A'$. We need to prove that 
\[
	(Z^- \circ \tau^\mA)(x') \leq \tau^\mAp(x'). \label{eq: HJWKA}
\]
Let $x$ be an arbitrary element of $A$. It is sufficient to show that  
\begin{equation}
Z(x,x') \fand \tau^\mA(x) \leq \tau^\mAp(x'). \label{eq: OEHQF}
\end{equation}
Since $Z_2$ is a fuzzy bisimulation between $\mM$ and~$\mMp$, by the condition~\eqref{eq: FB2} for $Z_2$, 
\[ \E y' \in \Delta^\mMp\ (Z_2(x,x') \fand \varrho^\mM(x,s_f) \leq \varrho^\mMp(x',y') \fand Z_2(s_f,y')). \]
By the assumptions about $Z_2$, $\mM$ and $\mM'$, this means that there exists $y' \in \Delta^\mMp$ such that  
\begin{equation}
Z(x,x') \fand \tau^\mA(x) \leq \varrho^\mMp(x',y') \fand Z_2(s_f,y'). \label{eq: RHSHQ}
\end{equation}
If $y' \neq s'_f$, then $Z_2(s_f,y') = 0$, and by~\eqref{fop: GDJSK 40}, \eqref{eq: RHSHQ} implies~\eqref{eq: OEHQF}. If $y' = s'_f$, then $\varrho^\mMp(x',y') = \tau^\mAp(x')$, $Z_2(s_f,y') = 1$ and \eqref{eq: RHSHQ} also implies~\eqref{eq: OEHQF}.

For the converse, suppose that the underlying residuated lattice is linear, $\mA$ and $\mAp$ are image-finite and $Z$ is a fuzzy bisimulation between~$\mA$ and~$\mAp$. We show that $Z_2$ is a fuzzy bisimulation between $\mM$ and~$\mMp$. Recall that $\SP = \{i,f\}$. Let $p \in \SP$ and $\varrho \in \SA$. 
We need to prove that the following assertions hold for all possible values of the free variables:
\begin{eqnarray}
\!\!\!\!\!\!\!\!\!\!&& Z_2(x,x') \leq (p^\mM(x) \fequiv p^\mMp(\red{x'})) \label{eq: FB1x} \\
\!\!\!\!\!\!\!\!\!\!&& \E y' \in \Delta^\mMp\ (Z_2(x,x') \fand \varrho^\mM(x,y) \leq \varrho^\mMp(x',y') \fand Z_2(y,y')) \label{eq: FB2x} \\
\!\!\!\!\!\!\!\!\!\!&& \E y \in \Delta^\mM\ (Z_2(x,x') \fand \varrho^\mMp(x',y') \leq \varrho^\mM(x,y) \fand Z_2(y,y')). \label{eq: FB3x}
\end{eqnarray}

Let $x \in \Delta^\mM$, $x' \in \Delta^\mMp$ and consider the assertion~\eqref{eq: FB1x}. 
If $\tuple{x,x'} \in A \times A'$, then $p^\mM(x) = p^\mMp(x') = 0$ and \eqref{eq: FB1x} clearly holds. 
If $Z_2(x,x') = 0$, then \eqref{eq: FB1x} also holds. 
Suppose that $\tuple{x,x'} \notin A \times A'$ and $Z_2(x,x') > 0$. 
Thus, $\tuple{x,x'} = \tuple{s_i,s'_i}$ or $\tuple{x,x'} = \tuple{s_f,s'_f}$. 
In both of these cases, for any $p \in \SP = \{i,f\}$, $p^\mM(x) = p^\mMp(x')$. Hence, \eqref{eq: FB1x} holds. 

Let $x,y \in \Delta^\mM$, $x' \in \Delta^\mMp$ and consider the assertion~\eqref{eq: FB2x}. 
If $Z_2(x,x') \fand \varrho^\mM(x,y) = 0$, then \eqref{eq: FB2x} clearly holds. Suppose that $Z_2(x,x') \fand \varrho^\mM(x,y) > 0$. By~\eqref{fop: GDJSK 40}, it follows that $Z_2(x,x') > 0$ and $\varrho^\mM(x,y) > 0$. 
There are the following cases.
\begin{itemize}
\item Case $x,y \in A$ and $x' \in A'$: We have 
\begin{equation}
Z_2(x,x') \fand \varrho^\mM(x,y) = Z(x,x') \fand \delta_\varrho^\mA(x,y). \label{eq: HFJSK 1} 
\end{equation}
By~\eqref{eq: HFHAJ 2}, 
\begin{equation}
Z(x,x') \fand \delta_\varrho^\mA(x,y) \leq (\delta_\varrho^\mAp \circ Z^-)(x',y). \label{eq: HFJSK 2}
\end{equation}
Since $\mAp$ is image-finite and the underlying residuated lattice is linear and complete, there exists $y' \in A'$ such that 
\begin{equation}
(\delta_\varrho^\mAp \circ Z^-)(x',y) = \delta_\varrho^\mAp(x',y') \fand Z(y,y') = \varrho^\mMp(x',y') \fand Z_2(y,y'). \label{eq: HFJSK 3} 
\end{equation}
The assertions~\eqref{eq: HFJSK 1}--\eqref{eq: HFJSK 3} together imply~\eqref{eq: FB2x}. 

\item Case $x = s_i$, $y \in A$ and $x' = s'_i$: We have 
\begin{equation}
Z_2(x,x') \fand \varrho^\mM(x,y) = 1 \fand \sigma^\mA(y) = \sigma^\mA(y). \label{eq: HFJSK 4} 
\end{equation}
By~\eqref{eq: HFHAJ 1}, 
\begin{equation}
\sigma^\mA(y) \leq (\sigma^\mAp \circ Z^-)(y). \label{eq: HFJSK 5}
\end{equation}
Since $\mAp$ is image-finite and the underlying residuated lattice is linear and complete, there exists $y' \in A'$ such that 
\begin{equation}
(\sigma^\mAp \circ Z^-)(y) = \sigma^\mAp(y') \fand Z(y,y') = \varrho^\mMp(x',y') \fand Z_2(y,y'). \label{eq: HFJSK 6} 
\end{equation}
The assertions~\eqref{eq: HFJSK 4}--\eqref{eq: HFJSK 6} together imply~\eqref{eq: FB2x}. 

\item Case $x \in A$, $y = s_f$ and $x' \in A'$: We have 
\begin{equation}
Z_2(x,x') \fand \varrho^\mM(x,y) = Z(x,x') \fand \tau^\mA(x). \label{eq: HFJSK 7} 
\end{equation}
By~\eqref{eq: HFHAJ 3}, 
\begin{equation}
Z(x,x') \fand \tau^\mA(x) \leq \tau^\mAp(x') = \varrho^\mMp(x',s'_f) \fand Z_2(y,s'_f). \label{eq: HFJSK 8}
\end{equation}
The assertions~\eqref{eq: HFJSK 7} and~\eqref{eq: HFJSK 8} together imply~\eqref{eq: FB2x}. 
\end{itemize}

The assertion~\eqref{eq: FB3x} can be proved analogously.
\myend
\end{proof}

\end{document}